\documentclass[11pt]{article}

\usepackage{amsfonts}
\usepackage{amsmath}
\usepackage{lipsum}
\usepackage{amssymb}
\usepackage{amsmath}
\usepackage{amsthm}
\usepackage{geometry}
\usepackage{multirow}
\usepackage{booktabs}
\usepackage{xr}
\usepackage{float}
\usepackage{rotating}

\usepackage[authoryear]{natbib}

\newcommand{\myref}[2]{\hyperref[#1]{#2}}
\numberwithin{equation}{section}

\usepackage{latexsym}
\usepackage{graphicx} 
\usepackage{enumerate}
\usepackage{setspace}
\usepackage[toc,title,titletoc,header]{appendix}
\usepackage[bottom]{footmisc} 
\usepackage[pdfborder={0 0 0}]{hyperref}
\hypersetup{
 colorlinks=true,linkcolor=magenta, citecolor=blue, filecolor=magenta, 
 linkbordercolor={0 1 1}, citebordercolor={1 0 0},
}

\usepackage{lineno}
\setpagewiselinenumbers

\allowdisplaybreaks

\usepackage{xcolor}
\newtheorem{theorem}{Theorem}
\newtheorem{lemma}{Lemma}
\newtheorem{corollary}{Corollary}
\newtheorem{example}{Example}
\let\oldexample\example
\let\endoldexample\endexample

\renewenvironment{example}
  {\pushQED{\qedsymbol}\oldexample}
  {\hfill\qedhere\endoldexample\popQED}
\newcounter{assumptionM}
\newcounter{assumptionA}
\newtheorem{inner_runningexample}{Running Example}
\newenvironment{runningexample}
  {\pushQED{\qedsymbol}\begin{inner_runningexample}}
  {\hfill\qedhere\end{inner_runningexample}\popQED}

\def\theassumptionM{M.\arabic{assumptionM}}
\def\theassumptionA{A.\arabic{assumptionA}}
\begin{document}
	\relax
	\hypersetup{pageanchor=false}
	\hypersetup{pageanchor=true}
{\author{
Federico A. Bugni\\
Department of Economics\\
Northwestern University\\
\url{federico.bugni@northwestern.edu}
\and
Joel L. Horowitz\\
Department of Economics\\
Northwestern University\\
\url{joel-horowitz@northwestern.edu}
\and
Linqi Zhang\\
Department of Decisions, Operations and Technology\\
The Chinese University of Hong Kong\\
\url{linqizhang@cuhk.edu.hk}
}}
 
\title{Demand estimation without knowledge of outside good shares\thanks{Corresponding author: \href{mailto:federico.bugni@northwestern.edu}{federico.bugni@northwestern.edu}.
We thank Gaston Illanes, Ivan Canay, Ruixuan Liu, and the seminar participants at the 2025 CUHK Econometrics Workshop for their comments and suggestions, which have improved the manuscript. Any errors are our own.}}

\maketitle

\vspace{-0.3in}
\thispagestyle{empty} 

\begin{spacing}{1.2}

\begin{abstract}
The BLP model is the workhorse framework for estimating demand for differentiated products using aggregate product shares. In practice, however, the share of the outside good is often unavailable. This paper studies identification and inference in the BLP model when the share of the outside good is unobserved. We show that the model is partially identified, and we derive the identified sets for the structural parameters and other quantities of economic interest. We also develop inference procedures based on moment inequalities that deliver valid confidence sets for these structural parameters and quantities of economic interest. We illustrate our results with an empirical application based on the tuna data analyzed by \citet{gandhi/lu/shi:2023}.

\end{abstract}
\end{spacing}



\thispagestyle{empty} 
\newpage

\section{Introduction}\label{sec:Introduction}

Discrete-choice demand models for differentiated products are a central tool in empirical industrial organization and are widely used in other applied fields in economics; see \cite{berry/haile:2014,berry/haile:2021} and references therein. Among these, the random-coefficients logit model introduced by \citet{berry:1994} and developed by \citet{berry/levinsohn/pakes:1995}, hereinafter referred to as BLP, has become the workhorse framework for estimating demand systems, conducting counterfactual analyses such as mergers, and evaluating market power, pass-through, and welfare.

The BLP model posits that industry-level demand arises from a discrete-choice model of consumer demand with unobserved heterogeneity. In each market, consumers choose among $J$ differentiated products and one {\it outside good}, representing the option of not choosing any of the $J$ products. The outside good is a key ingredient of the BLP model. As \citet[p.\ 14]{berry/haile:2021} explain, ``In a discrete choice model without an outside good, the market demand elasticity would always be zero.''
According to \citet[p.\ 2088]{berry/gandhi/haile:2013}, the interpretation of the outside good can vary across applications: it may represent a genuine economic alternative, a reference alternative relative to which utilities are normalized, or a purely artificial technical device. 

A central feature of BLP is that it allows identification and estimation of the structural demand system using aggregate product-level market shares. In the standard formulation, these shares include both the shares of the $J$ inside goods and the share of the outside good. In many empirical applications, however, the share of the outside good is not observed. This situation arises naturally when only aggregate shares of the $J$ inside goods are observed, or when aggregate purchase data are available but the size or boundaries of the relevant market are not well defined. 
To address this issue, the applied literature often relies on ad hoc assumptions regarding the outside good or, equivalently, market size. A recent survey by \cite{Zhang:2024} finds that, among 29 top-five publications using the BLP model, 24 rely on ad hoc assumptions about market size or the outside good, while only 5 assess the robustness of their results to these assumptions.

This paper adopts an agnostic approach to the share of the outside good, or equivalently, to market size. We consider a setting in which the researcher can credibly restrict the share of the outside good in each market to a set $S_0(w)$, which can depend on the observed market data $w$. Our framework accommodates varying degrees of missing information about the outside option. At one extreme is the fully agnostic case in which $S_0(w) = (0,1)$. At the other extreme, the existing literature treats the outside share as observed, possibly up to sampling error, effectively assuming that $S_0(w)$ is a singleton. Intermediate cases allow the researcher to work with incomplete specifications of the outside share. For example, the researcher may know only that the outside share is bounded below by some constant $c\in(0,1)$, without knowing its exact value. This corresponds to setting $S_0(w)=[c,1)$. Alternatively, if the researcher has a benchmark value for the outside share, sensitivity analysis can be conducted by specifying $S_0(w)$ as a neighborhood around that value.

We study identification and inference in the BLP model when the share of the outside good is restricted to the set $S_0(w)$. Our first contribution is to characterize the identified set for the structural parameters of the BLP model in this context. Our analysis shows that once the share of the outside good is no longer pinned down to a single value, the BLP model is typically partially identified, even under otherwise standard assumptions. Using tools from random set theory (\cite{molchanov:2005,molchanov/molinari:2018}), we derive an explicit characterization of the identified set that fully exploits the structure of the BLP demand system. In essence, our identified set is characterized by an infinite collection of conditional moment inequalities.

Our second contribution is to characterize the identified sets of economically relevant quantities in the BLP model, such as demand shocks, demand elasticities, equilibrium markups, diversion ratios, and market shares. We show how partial identification of the structural parameters propagates to these quantities of interest, and we derive sharp identified sets for them. 

To illustrate these results, we present a numerical example that shows how the identified sets of parameters and economic quantities vary under different specifications of $S_0(w)$. The results demonstrate that although small departures from a singleton outside share lead to partial identification of the structural parameters, the resulting sets remain sufficiently informative to support economically meaningful conclusions.

Our third contribution is to develop inference procedures for the BLP model in this partially identified setting. Building on the methods of \citet{andrews/shi:2013} and \citet{chernozhukov/chetverikov/kato:2019}, and using our characterization of the identified set, we transform the conditional moment inequalities into a (potentially large) collection of unconditional moment inequalities and construct confidence sets that are asymptotically valid uniformly over a broad class of data-generating processes. From here, we show how to conduct inference on both the structural parameters and quantities of economic interest in the BLP model.

The remainder of the paper is organized as follows. Section \ref{sec:LitReview} reviews the related literature and places our contribution in context. Section \ref{sec:standard_BLP} briefly reviews the standard BLP model with observed outside shares. Section \ref{sec:Identification} contains the paper’s main identification results, deriving identified sets for the structural parameters and economically relevant quantities when outside shares are unobserved. Section \ref{sec:Inference} builds on these results to construct confidence sets. Section \ref{sec:Application} illustrates the methodology in an empirical application, and Section \ref{sec:Conclusions} concludes. Proofs and auxiliary results are collected in the Appendix.

\subsection{Literature review}\label{sec:LitReview}

This paper contributes to the literature on differentiated-products demand estimation by studying identification and inference in the BLP model after relaxing the standard assumption that the outside share, or equivalently market size, is known. In standard applications of this model, the researcher assumes that the share of the outside good is observed or can be constructed from the available data. When this share is unobserved, the existing literature typically imposes additional structure to substitute for the missing information. Without direct data on the outside option, however, these restrictions are difficult to assess empirically. This paper instead studies what can be learned without imposing parametric or auxiliary structure that point identifies the outside share.

One strand of the literature incorporates a parametric specification for the outside shares within the BLP framework and estimates it jointly with demand parameters. \citet{Berry:2006tt} introduce additional structure that links the market size, defined as the sum of inside and unobserved outside quantities, to observed covariates. More recently, \citet{Zhang:2024} models the market size as a function of observed market-level variables, known up to a finite-dimensional parameter vector, which is estimated jointly with demand. Under suitable conditions, this restores point identification using the same data requirements as in standard BLP estimation. Relatedly, \citet{Huang:2013vh} study a standard logit model without random coefficients and use market fixed effects to eliminate bias from unobserved outside shares. They recover the size of the outside option under the additional assumption that market size is constant across markets. In contrast, our approach does not seek to restore point identification through additional parametric restrictions; instead, it allows the model to be set-identified.

A second strand of the literature relies on supply-side conditions and auxiliary information, such as cost or accounting data, to estimate market size (e.g., \cite{Chu:2011ul}, \cite{Byrne:2022vz}, and \cite{Kim:2024ti}). These approaches exploit equilibrium pricing conditions to infer market size and, implicitly, the outside share. Identification in these settings, therefore, depends on additional structural assumptions on the supply side. By contrast, we abstract from supply-side information when identifying the structural parameters and study what can be learned from the demand side alone.

A third approach, increasingly common in applied work, predicts total demand, or proxies for it, in a first step. In grocery and airline applications, for example, researchers often predict total store trips or arriving passengers using stand-alone regressions such as gravity equations or Poisson arrival processes (e.g., \cite{Sweeting:2020tw}, \cite{Li:2022tk}, \cite{Hortacsu:2022wt}, and \cite{Backus:2021aa}). However, total store trips or arriving passengers do not capture consumers who choose not to visit the store or who use alternative modes of transportation. As a result, the estimated quantity still needs to be scaled by an ad hoc factor to account for these unobserved outside options.

Our paper is also related to the literature on inference in differentiated product demand models. There is a large body of literature addressing this problem under the assumption that inside and outside market shares are observed. \cite{berry/linton/pakes:2004} and \cite{armstrong:2016} study asymptotics as the number of products grows, whereas \citet{freyberger:2015} and \cite{hong/li/li:2021} consider asymptotics in the number of markets, as we do. The main difference between these papers and ours is that we treat outside shares as missing data and address the associated identification and inference problem.

\subsection{The standard BLP model}\label{sec:standard_BLP}

This paper considers the BLP model when the share of the outside good is unobserved. Before delving into this problem, we briefly review the standard BLP model, in which all shares are observed. We follow the notation and the standard assumptions imposed in \cite{berry:1994,berry/levinsohn/pakes:1995,berry/haile:2021}.

We observe an i.i.d.\ sample of markets $m=1,2,\dots ,n$. Market observations typically correspond to different cities and/or periods. That is, $m$ is determined by the pair $(t,c)$, where $t=1,2,\dots ,T_{c}$ denotes periods observed for city $c$ and $c=1,2,\dots ,C$ denotes cities observed. Since markets are assumed to be i.i.d., we can drop the index $m$ from the identification analysis. 

In each market, there is one ``outside'' product (labeled as zero) and $J$ ``inside''  products. We assume that $J$ is a fixed, known number. Throughout this section, we observe the shares of all products, conditional on market characteristics. These are determined by the decisions of a representative consumer in a standard BLP discrete-choice model. We assume that there is a unit mass of consumers, indexed by $i \in \mathcal{I}$, each choosing a single product.

For each market, the observed random variables are $( s,x,p,z) $, where $s=( s_{j}) _{j=1}^{J}\in (0,1)^{J}$ denotes the shares of the inside products, $x=( x_{j}') _{j=1}^{J}$ with $x_{j}\in \mathbb{R}^{d_{X}}$ denotes the non-price inside product characteristics, $p=( p_{j}) _{j=1}^{J}$ denotes the prices of the inside products, and $z$ denotes a vector of instruments.  In the canonical case of the car market, $j=1,2,\dots, J$ is a list of the available cars, $x$ would represent car features like fuel efficiency, size, and a constant vector, and $z$ would represent cost shifters and characteristics of other products, which generate changes in the observables but are assumed to be mean independent of the unobserved demand shock. The outside share can be deduced from the inside shares: $s_{0}=1-\sum_{j=1}^{J}s_{j}$. By definition, $( s_{0},s) \in \Delta _{J+1}$, where $\Delta _{\ell}$ for $\ell \in \mathbb{N}$ denotes the $\ell$-dimensional simplex.

The indirect utility of the (inside) good $j=1,2,\dots ,J$ for the representative individual $i$ is 
\begin{align*}
u_{ij}~ =~\tilde{\beta}_{i} x_{j}- \tilde{\alpha}_{i}p_{j}+\xi _{j}+\epsilon _{ij} 
~ =~\beta x_{j} +\zeta _{i} x_{j}-\alpha p_{j}-\nu _{i} p_{j}+\xi _{j}+\epsilon _{ij},
\end{align*}
where $\tilde{\beta}_{ik}=(\beta _{k}+\zeta _{ik})$ is the random coefficient on $x_{jk}$, $\beta =( \beta _{k}) _{k=1}^{d_{X}}$ is the vector of mean level of taste parameter for $x_{k}$, $\zeta _{ik}$ is a random component of the taste parameter for $x_{k}$, $\zeta _{i}  = (\zeta _{ik}) _{k=1}^{d_{X}}$ is the associated vector, $\tilde{\alpha}_{i}=(\alpha +\nu _{i})$ is the random coefficient on $p_{j}$, $\alpha $ is the mean level of distaste parameter for $p_{j}$, $\nu _{i}$ is the random component of the distaste parameter for $p_{j}$, $\xi _{j}$ is the product-specific demand shock, and $\epsilon _{ij}$ is the individual-specific taste shock. In the case of the cars, $\beta _{k}$ represents the mean taste for the observed characteristic $k$, $\zeta _{ik}$ represents customer $i$'s preference shock for this characteristic, $\nu _{i}$ represents customer $i$'s personal sensitivity to price, and $\xi _{j}$ could represent product $j$'s brand prestige. The indirect utility of the outside good for the representative individual $i$ is $u_{i0}=\epsilon _{i0}$.

In this model, the following objects are random and unobserved to the econometrician: the individual-specific taste shocks $((\epsilon _{ij})_{j=0}^{J})_{i \in \mathcal{I}}$, the random components $(\zeta _{i},\nu _{i})_{i \in \mathcal{I}}$, and the demand shocks $\xi = (\xi _{j})_{j=1}^{J}$. The standard BLP model imposes the following conditions on these shocks. First, $((\epsilon _{ij})_{j=0}^{J})_{i \in \mathcal{I}}$ is independent of the data and other shocks, and i.i.d.\ Type I extreme value.\footnote{The Type I extreme value assumption is commonly adopted in empirical work, but it is not essential for our identification analysis. It can be relaxed, provided the connected substitutes condition in \cite{berry/gandhi/haile:2013} holds, which ensures the demand system is invertible under full information.} Second, $(\zeta _{i},\nu _{i})_{i \in \mathcal{I}}$ are also assumed independent of the data and other shocks, and i.i.d.\ with a density function $f(\zeta,\nu;\lambda ) $ that is known up to a parameter $\lambda \in \Lambda \subset \mathbb{R}^{d_{\lambda }}$.\footnote{In practice, $f(\zeta,\nu;\lambda ) $ is often assumed to be normally distributed with mean zero and variance–covariance matrix $\lambda \in \Lambda $, and so $\Lambda$ represents a space of variance-covariance matrices. In this case, setting $\lambda$ to the zero matrix implies $\zeta=\mathbf{0}_{d_X}$ and $\nu=0$, which yields the plain logit model.} Finally, the demand shocks are $\xi = (\xi _{j})_{j=1}^{J}$, and assumed to satisfy the IV exogeneity condition:
\begin{equation}
    E[ \xi _{j}|z] ~=~0~~~~ \text{for}~j=1,\dots,J. 
    \label{eq:IVcondition}
\end{equation}
The structural parameters of the model are: 
\begin{equation*}
\theta ~=~( \alpha ,\beta ,\lambda ) ~\in ~\Theta ~\subset ~\mathbb{R}\times \mathbb{R}^{d_{X}}\times \Lambda,
\end{equation*}
where $\Theta$ denotes the parameter space for $\theta$.


For any product $j=1,\dots,J$, this model implies the following conditional market share:
\begin{equation}
P(y=j|x,p,\xi ;\theta )~=~\int_{(\zeta ,\nu )}\frac{\exp ( \beta x_{j} -\alpha p_{j}+\xi _{j}+\zeta x_{j} -\nu p_{j}  ) }{1+\sum_{b=1}^{J}\exp (\beta  x_{b} -\alpha p_{b}+\xi _{b}+\zeta x_{b} -\nu p_{b} ) } f( \zeta ,\nu ;\lambda ) d( \zeta ,\nu ) .
\label{eq:BLP_probs_1}
\end{equation} 
See Lemma \ref{lem:BLP_probs} in the appendix for a derivation. Given that each market contains a continuum of customers, we have that, conditional on $(x,p,\xi)$,
\begin{equation}
s ~=~ \big(P(y=j|x,p,\xi ;\theta )\big)_{j=1}^{J}.
\label{eq:sharesUncond}
\end{equation}

Under mild assumptions in \cite{berry:1994} and \cite{berry/levinsohn/pakes:1995}, this model is point identified. We now briefly outline the argument. 
For any $\delta  \in \mathbb{R}^{J}$, let the function $ \sigma (\delta ,x,p;\lambda ):\mathbb{R}^{J}\times \mathbb{R}^{d_{X}}\times \mathbb{R}^{J}\times \mathbb{R}^{d_{\lambda }}\to (0,1)^{J}$ be defined as follows:
\begin{equation}
\sigma (\delta ,x,p;\lambda )~=~\left(\int_{(\zeta ,\nu )}\frac{   \exp ( \delta _{j}+\zeta x_{j} - \nu p_{j} ) }{1+\sum_{b=1}^{J}\exp ( \delta _{b}+\zeta x_{b} -\nu p_{b} ) }f( \zeta ,\nu ;\lambda ) d( \zeta ,\nu )\right)_{j=1}^{J}.
\label{eq:sigma_defn}
\end{equation}
It is also convenient to define the corresponding outside-share function, $ \sigma_0 (\delta,x,p;\lambda ):\mathbb{R}^{J}\times \mathbb{R}^{d_{X}}\times \mathbb{R}^{J}\times \mathbb{R}^{d_{\lambda }}\to (0,1)$, given by
\begin{equation}
\sigma_0 (\delta ,x,p;\lambda )~=~1-\sum_{j=1}^J\sigma_j (\delta ,x,p;\lambda )~=~\int_{(\zeta ,\nu )}\frac{   1}{1+\sum_{b=1}^{J}\exp ( \delta _{b}+\zeta x_{b} -\nu p_{b} ) }f( \zeta ,\nu ;\lambda ) d( \zeta ,\nu ).
\label{eq:sigma0_defn}
\end{equation}
\cite{berry:1994} and \cite{berry/levinsohn/pakes:1995} show that, for any fixed $(s,x,p,\lambda)$, the mapping $\delta \to \delta + (\ln(s_j) - \ln{\sigma_j (\delta ,x,p;\lambda )})_{j=1}^{J}$ is a contraction in the sup-norm. Consequently, for any vector of market shares generated by this model, there exists a unique value of $\delta$ that solves $\sigma (\delta,x,p;\lambda ) = s$. By combining this with \eqref{eq:BLP_probs_1} and \eqref{eq:sharesUncond}, we obtain that
\begin{equation}
\xi_j  ~=~ \sigma^{-1}_j( s,x,p;\lambda ) - \beta x_j +\alpha p_j~~~\text{for}~j=1,\dots,J,
\label{eq:BLP_shocks}
\end{equation}
where $\sigma^{-1}_j$ denotes the $j$th coordinate of  $\sigma^{-1}( s,x,p;\lambda )$. This and \eqref{eq:IVcondition} yield the following equation:
\begin{equation}
E[\sigma^{-1}_j( s,x,p;\lambda ) - \beta x_j + \alpha p_j|z] ~=~0~~~\text{for}~j=1,\dots,J.
\label{eq:BLP_IVGMM}
\end{equation}
Under standard conditions regarding the variability of the instruments, \eqref{eq:BLP_IVGMM} point identifies the parameter of interest $\theta =(\alpha,\beta,\lambda )$. In particular, \citet[Page 23]{berry/haile:2021} argue that one can use \eqref{eq:BLP_IVGMM} to produce unconditional moment conditions that point identify $\theta$. Estimation and inference then proceed by standard GMM methods. 


Throughout the paper, we use the special case of the model without random coefficients, i.e., the plain logit model, as a running example to illustrate the main results and derivations. 

\begin{runningexample}[Identification in the plain logit model]\label{ex:Id_plainLogit}
Let $\bar{\lambda}\in\Lambda$ be the parameter value that makes the distribution of random coefficients $(\zeta,\nu)$ degenerate at zero, i.e.,
\begin{equation*}
f(\zeta,\nu;\bar{\lambda}) ~=~ I[\zeta=\mathbf{0}_{d_X},~ \nu=0].
\end{equation*}
We can specialize the BLP model to the plain logit case by assuming that $\Lambda=\{\bar{\lambda}\}$, leading to a parameter space $\Theta = \mathbb{R}\times \mathbb{R}^{d_{X}}\times \{\bar\lambda\}$. The only parameters left to identify are $\alpha$ and $\beta$.

In this special setting, the conditional choice probabilities are given by
\begin{equation}
P(y=j|x,p,\xi;\theta) ~=~ \frac{\exp(\beta x_j - \alpha p_j + \xi_j)}{1+\sum_{b=1}^J \exp(\beta x_b - \alpha p_b + \xi_b)} ~~~\text{for}~j=1,\dots,J,
\label{eq:BLP_probs_logit}
\end{equation}
and \eqref{eq:sigma_defn} and \eqref{eq:sigma0_defn} become
\begin{align*}
\sigma (\delta ,x,p;\lambda)~=~\left(\frac{ \exp ( \delta _{j}) }{1+\sum_{b=1}^{J}\exp ( \delta _{b}) }\right)_{j=1}^J~~\text{and}~~
\sigma_0 (\delta ,x,p;\lambda)~=~\frac{1 }{1+\sum_{b=1}^{J}\exp ( \delta _{b}) }.
\end{align*}
The function $\sigma(\delta,x,p;\lambda)$ is invertible in the first argument. Using $s_0 = 1-\sum_{j=1}^J s_j$, we obtain
\begin{equation}
\sigma^{-1}(s,x,p;\lambda)~=~ (\ln ( s_{j}/s_0))_{j=1}^J,
\label{eq:sigmaj_inv}
\end{equation}
where $\sigma ^{-1}$ is the inverse of $\sigma(\delta ,x,p;\lambda)$ with respect to the first argument.
Substituting into \eqref{eq:BLP_IVGMM}, we obtain the following conditional moment condition:
\begin{equation}
E[\ln ( s_{j}/s_0) - \beta x_{j}  + \alpha p_{j}|z] ~=~0~~~~ \text{for }j=1,\dots,J. 
\label{eq:BLP_IVGMM_logit}
\end{equation}
Provided that the instrument $z$ has sufficient variability, \eqref{eq:BLP_IVGMM_logit} identifies $(\alpha,\beta)$.
\end{runningexample}

In addition to the structural parameters $\theta$, researchers who use the BLP model are often interested in other objects within this framework, such as demand elasticities, markups, and diversion ratios. 
For any pair of products $j,k=1,\dots,J$, the elasticity of demand for product $j$ with respect to the price of product $k$ is given by
\begin{equation*}
e_{jk} ~=~\frac{\partial \ln P( y=j|x,p,\xi ;\theta ) }{\partial \ln p_{k}}.
\end{equation*}
By combining this with \eqref{eq:sharesUncond} and \eqref{eq:sigma_defn}, we obtain an expression for the own-price elasticity of product $j$, given by
\begin{equation}
e_{jj} ~=~-\frac{p_{j}}{s_{j}}\int_{(\zeta ,\nu )}(\alpha +\nu )
\left(\begin{array}{c}
\tfrac{\exp (\beta x_{j} -\alpha p_{j}+\xi _{j}+\zeta x_{j} -\nu p_{j})}{1+\sum_{b=1}^{J}\exp (\beta x_{b} -\alpha p_{b}+\xi _{b}+\zeta x_{b} -\nu p_{b})}\times \\
\left( 1-\tfrac{\exp (\beta x_{j} -\alpha p_{j}+\xi _{j}+\zeta x_{j} -\nu p_{j})}{1+\sum_{b=1}^{J}\exp (\beta x_{b} -\alpha p_{b}+\xi _{b}+\zeta x_{b} -\nu p_{b})} \right)
\end{array}\right)
f(\zeta ,\nu ;\lambda )d(\zeta ,\nu ),
\label{eq:elasticity_own}
\end{equation}
and an expression for the cross-price elasticity of product $j$ with respect to the price of product $k$, given by
\begin{equation}
e_{jk} ~=~\frac{p_{k}}{s_{j}}\int_{(\zeta ,\nu )}(\alpha +\nu )
\left(\begin{array}{c}
\tfrac{\exp (\beta x_{j} -\alpha p_{j}+\xi _{j}+\zeta x_{j} -\nu p_{j})}{1+\sum_{b=1}^{J}\exp (\beta x_{b} -\alpha p_{b}+\xi _{b}+\zeta x_{b} -\nu p_{b})}\times \\
\tfrac{\exp (\beta x_{k} -\alpha p_{k}+\xi _{k}+\zeta x_{k} -\nu p_{k})}{1+\sum_{b=1}^{J}\exp (\beta x_{b} -\alpha p_{b}+\xi _{b}+\zeta x_{b} -\nu p_{b})}
\end{array}\right)
f(\zeta ,\nu ;\lambda )d(\zeta ,\nu ).
\label{eq:elasticity_cross}
\end{equation}
The markup of product $j$ is the difference between $p_{j}$ and the marginal cost of product $j$. Under the assumption that single-product firms choose prices to maximize profits and that these are unconstrained in equilibrium, we obtain an expression relating the markup and the own elasticity:
\begin{equation}
M_{j} ~=~\frac{-P( y=j|x,p,\xi ;\theta ) }{\partial P( y=j|x,p,\xi ;\theta ) /\partial p_{j}} ~=~ - \frac{p_j}{e_{jj}}.
\label{eq:markup}
\end{equation}
Finally, the diversion ratio between products $j$ and $k$ quantifies the extent to which purchases shift from product $k$ to product $j$ in response to an infinitesimal increase in the price of product $k$. Formally, it is given by
\begin{equation*}
    D_{jk} ~=~ - \frac{\partial P( y=j|x,p,\xi ;\theta )}{\partial p_{k}} \big/ \frac{\partial P( y=k|x,p,\xi ;\theta )}{\partial p_{k}}.
\end{equation*}
By combining this with \eqref{eq:sharesUncond} and \eqref{eq:sigma_defn}, we obtain the following expression:
\begin{equation}
D_{jk} ~=~ \frac{\int_{(\zeta ,\nu )}(\alpha + \nu) 
\left(\begin{array}{c}
\tfrac{\exp (\beta x_{j} -\alpha p_{j}+\xi _{j}+\zeta x_{j} -\nu p_{j})}{1+\sum_{b=1}^{J}\exp (\beta x_{b} -\alpha p_{b}+\xi _{b}+\zeta x_{b} -\nu p_{b})}\times \\
\left( \tfrac{\exp (\beta x_{k} -\alpha p_{k}+\xi _{k}+x_{k}\zeta -\nu p_{k})}{1+\sum_{b=1}^{J}\exp (\beta x_{b} -\alpha p_{b}+\xi _{b}+\zeta x_{b} -\nu p_{b})} \right)
\end{array}\right)
f(\zeta,\nu; \lambda) d(\zeta,\nu)}{\int_{(\zeta ,\nu )}(\alpha + \nu) 
\left(\begin{array}{c}
\tfrac{\exp (\beta x_{k} -\alpha p_{k}+\xi _{k}+x_{k}\zeta -\nu p_{k})}{1+\sum_{b=1}^{J}\exp (\beta x_{b} -\alpha p_{b}+\xi _{b}+\zeta x_{b} -\nu p_{b})}\times \\
\left( 1-\tfrac{\exp (\beta x_{k} -\alpha p_{k}+\xi _{k}+x_{k}\zeta -\nu p_{k})}{1+\sum_{b=1}^{J}\exp (\beta x_{b} -\alpha p_{b}+\xi _{b}+\zeta x_{b} -\nu p_{b})} \right)
\end{array}\right)
f(\zeta,\nu;\lambda)d(\zeta,\nu)}
\label{eq:diversion}
\end{equation}

The expressions in \eqref{eq:elasticity_own}-\eqref{eq:diversion} reveal that the elasticities, markups, and diversion ratios are functions of the structural parameters $\theta$, the observables $(s,p,x)$, and the demand shocks $\xi$.

\begin{runningexample}[Elasticities and markups in the plain logit model]
In the plain logit case (i.e., $\Lambda =\{\bar{ \lambda}\}$), the expressions in \eqref{eq:elasticity_own}-\eqref{eq:diversion} simplify considerably. For any $j,k=1,\dots,J$ with $j\neq k$, we obtain
\begin{align*}
e_{jj}~=~-\alpha p_{j}( 1-s_{j}),~~~e_{jk}~ =~\alpha p_{k}s_{k},~~~M_{j}~ =~1/(\alpha (1-s_{j})),~~\text{and}~~D_{jk} ~=~ {s_j}/(1-s_{k}).
\end{align*}
Notably, these expressions depend only on $p$, $s$, and $\alpha$, and the dependence on the other structural parameters or $\xi$ drops out.
\end{runningexample}

\section{Identification without knowledge of outside-good shares}\label{sec:Identification}

This section derives the sharp identified set for the BLP model when the outside-good share $s_0$ is unobserved. Apart from this, the structure of the model is identical to that described in Section \ref{sec:standard_BLP}. In particular, we observe an i.i.d.\ sample of markets indexed by $m=1,\dots, n$. Since markets are i.i.d., we suppress the market subscript and conduct the identification analysis using a representative market. 

Within each market, the observed random variables are $w=(\tilde{s},x,p,z)$, where
\begin{equation*}
    \tilde{s}~=~\Big( s_j \big/ \sum\nolimits_{b=1}^{J} s_b \Big)_{j=1}^{J}
    ~\in~    \Delta_J ,
\end{equation*}
and $s$, $x$, $p$, $z$, and $\Delta_J$ are as defined in Section \ref{sec:standard_BLP}. Unlike in the standard case, however, the outside-good share $s_0$ cannot be recovered from the observed data.

Throughout our analysis, we assume that the unobserved outside share $s_0$ belongs to a known, possibly data-dependent measurable set $S_0(w)\subseteq(0,1)$. We interpret $S_0(w)$ as the set of outside-good shares that are consistent with the researcher's information for a given realization of the observed data. We assume that $S_0(w)$ is closed relative to $(0,1)$.\footnote{This means that if $s_{0,b} \in S_0(w)$ for all $b\in \mathbb{N}$ and $s_{0,b} \to s_0 \in (0,1)$, then $s_0 \in S_0(w)$.} The fully agnostic case corresponds to $S_0(w)=(0,1)$, where $s_0=0$ and $s_0=1$ are ruled out by the specification of the BLP model. If the researcher has additional information about the share of the outside good, this information can be incorporated by restricting $S_0(w)$ to a narrower set.



Under our assumptions, we have that, conditional on $(x,p,\xi)$,
\begin{equation}
\tilde{s} ~=~ \big( P(y=j \mid x,p,\xi,y > 0;\theta)\big)_{j=1}^{J}.
\label{eq:sharesCond}
\end{equation}
By combining \eqref{eq:sharesCond} with the conditional choice probabilities in \eqref{eq:BLP_probs_1}, we obtain
\begin{equation}
\tilde{s} ~=~ \left(\frac{\int_{(\zeta ,\nu )}\frac{
\exp ( \beta x_{j}-\alpha p_{j}+ \xi_{j} + \zeta x_{j}  -\nu p_{j} )  }{ 1+\sum_{b=1}^{J}\exp ( \beta x_{b}-\alpha p_{b}+\xi _{b}+\zeta x_{b}  -\nu p_{b} ) 
}f( \zeta ,\nu ;\lambda ) d( \zeta ,\nu ) 
}{
\int_{(\zeta ,\nu )}\frac{\sum_{l=1}^{J}\exp ( \beta x_{l}-\alpha p_{l}+\xi _{l}+\zeta x_{l}  -\nu p_{l} ) }{ 1+\sum_{b=1}^{J}\exp ( \beta x_{b}-\alpha p_{b}+\xi _{b}+\zeta x_{b} - \nu p_{b} ) }f( \zeta ,\nu ;\lambda ) d( \zeta ,\nu) }\right)_{j=1}^{J}.
\label{eq:sharesCond2}
\end{equation}

For any $\delta  \in \mathbb{R}^{J}$, the analog of the function $\sigma (\delta ,x,p;\lambda )$ in \eqref{eq:sigma_defn} in the current context is $ \tilde\sigma (\delta ,x,p;\lambda ):\mathbb{R}^{J}\times \mathbb{R}^{d_{X}}\times \mathbb{R}^{J}\times \Lambda \to \Delta_J$, given by
\begin{equation}
\tilde\sigma (\delta ,x,p;\lambda )~=~\left(\frac{\int_{(\zeta ,\nu )}\frac{ \exp ( \delta _{j}+\zeta x_{j} -\nu p_{j} )   }{ 1+\sum_{b=1}^{J}\exp ( \delta _{b}+\zeta x_{b}  -\nu p_{b} ) 
}f( \zeta ,\nu ;\lambda ) d( \zeta ,\nu ) 
}{
\int_{(\zeta ,\nu )}\frac{\sum_{l=1}^{J}\exp ( \delta _{l}+\zeta x_{l}  -\nu p_{l} ) }{ 1+\sum_{b=1}^{J}\exp ( \delta _{b}+\zeta x_{b} -\nu p_{b} ) }f( \zeta ,\nu ;\lambda ) d( \zeta ,\nu) } \right)_{j=1}^{J}.
\label{eq:sigma_defn2}
\end{equation}
Unlike $\sigma(\delta ,x,p;\lambda )$, $\tilde\sigma(\delta ,x,p;\lambda)$ is not invertible in the first argument. The next example illustrates this point in the special case of the plain logit model.

\begin{runningexample}[Non-invertibility of $\tilde\sigma(\delta ,x,p;\lambda)$ in the plain logit model]
In the plain logit case (i.e., $\Lambda =\{\bar{ \lambda}\}$), \eqref{eq:sigma_defn2} becomes
\begin{equation}
\tilde{\sigma}(\delta ,x,p;\lambda )~=~\left(\frac{ \exp \left( \delta _{j}\right) }{\sum_{b=1}^{J}\exp ( \delta _{b}) }\right) _{j=1}^{J}.
\end{equation}
For any $C \in \mathbb{R}$, note that $\tilde{\sigma}(\delta ,x,p;\lambda)= \tilde{\sigma}(\delta +{\bf 1}_{J \times 1} C,x,p;\lambda)$, and so $\tilde{\sigma}(\delta ,x,p;\lambda)$ is not invertible in the first argument.
\end{runningexample}

Therefore, the argument used in Section \ref{sec:standard_BLP} to identify the parameters of the model does not apply. In fact, the next illustration shows that the model is not (point) identified.

\begin{runningexample}[Lack of identification in the plain logit model]\label{ex:NoId}
Consider the plain logit case (i.e., $\Lambda =\{\bar{ \lambda}\}$), \eqref{eq:sigma_defn2} with $x_{j}=(1,\tilde{x}_{j}')'$ and $\beta =(\beta _{1},\beta _{2})$. Then, \eqref{eq:sharesCond2} becomes
\begin{equation*}
\tilde{s}~=~\left(\frac{ \exp ( \beta _{1}+\beta _{2}\tilde{x} _{j}-\alpha p_{j}+\xi _{j}) }{\sum_{b=1}^{J}\exp ( \beta _{1}+\beta _{2}\tilde{x}_{b}-\alpha p_{b}+\xi _{b}) }\right)_{j=1}^{J}.
\end{equation*}

First, note that $\beta _{1}$ is not identified. To see this, note that for any $( \alpha ,( \beta _{1},\beta _{2}) ,\bar{\lambda}) \in \Theta $, $( \alpha ,( \beta _{1},\beta _{2}) ,\bar{\lambda}) $ and $( \alpha ,( \beta _{1}+C,\beta _{2}) ,\bar{ \lambda}) $ are observationally equivalent. 

Next, we argue that $( \beta _{2},\alpha ) $ is identified under mild additional conditions. To see this, note that
\begin{equation*}
\ln \left( \tilde{s}_{j}/\tilde{s}_{k}\right) ~=~\beta _{2}\left( \tilde{x} _{j}-\tilde{x}_{k}\right) -\alpha (p_{j}-p_{k})+(\xi _{j}-\xi _{k})~~~ \text{for }j,k=1,\dots,J.
\end{equation*}
Then,
\begin{equation*}
E\left[ \ln \left( \tilde{s}_{j}/\tilde{s}_{k}\right) |z\right] ~=~\beta _{2}E \left[ \tilde{x}_{j}-\tilde{x}_{k}|z\right] -\alpha E\left[ p_{j}-p_{k}|z \right] ~~~ \text{for }j,k=1,\dots,J.
\label{eq:fixed_linear_model}
\end{equation*}
Provided that $E[ \tilde{x}_{j}-\tilde{x}_{k}|z] $ and $E [ p_{j}-p_{k}|z] $ vary sufficiently for some $j,k$, $( \beta _{2},\alpha ) $ is identified.
\end{runningexample}

Example \ref{ex:NoId} shows that, when outside-good shares are unobserved, the fixed-coefficient version of the BLP model need not be point identified. It also enables us to characterize the identified set of structural parameters in this special case. The next section addresses the corresponding identification problem in the general BLP model with random coefficients.

\subsection{Identification of the structural parameters}

The following result provides a sharp characterization of the identified set of the BLP model. It derives the sharp identified set for the structural parameters using tools from random set theory.

\begin{theorem}[Identified set of the structural parameters]\label{thm:Id_set_expression}
Assume that the conditional distribution of $\{(\tilde{s},x,p)|z\}$ is non-atomic a.e.\ $z\in \mathbb{S}_{z}$,
and that, for any $\theta =( \alpha ,\beta ,\lambda ) \in \Theta$, the following random set is integrable:
\begin{equation}
\mathcal{U}( w;\theta ) =\left\{ \xi \in \mathbb{R}^{J}:
\begin{array}{c}
\tilde{s}~=~\left( \dfrac{ \int_{(\zeta ,\nu )}\frac{ \exp ( \beta x_{j}-\alpha p_{j} + \xi_{j} + \zeta x_{j} -\nu p_{j} ) }{ 1+\sum_{b=1}^{J}\exp ( \beta x_{b}-\alpha p_{b} + \xi_{b} +\zeta x_{b}  -\nu p_{b} ) }f( \zeta ,\nu ;\lambda ) d( \zeta ,\nu ) }{\int_{(\zeta ,\nu )}\frac{\sum_{l=1}^{J}\exp ( \beta x_{l}-\alpha p_{l} + \xi_{l} +\zeta x_{l}-  \nu p_{l}) }{1+\sum_{e=1}^{J}\exp ( \beta x_{e}-\alpha p_{e}+\xi _{e}+\zeta x_{e}  -p_{e}\nu ) }f( \zeta ,\nu ;\lambda ) d( \zeta ,\nu ) } \right) _{j=1}^{J},\\
\sigma_0\big((\beta x_j-\alpha p_j+\xi_j)_{j=1}^J,x,p;\lambda\big)\in S_0
\end{array}
\right\}.
\label{eq:U_defn}
\end{equation}
Then, the identified set of $\theta =( \alpha ,\beta ,\lambda ) $ is
\begin{equation}
\Theta _{I}(P) =\left\{ 
\begin{array}{c}
( \alpha ,\beta ,\lambda ) \in \Theta: \\ 
E\Big[\sup_{s_{0}\in S_0(w)}v^{\prime }\big( \sigma ^{-1}( \tilde{s}(1-s_{0}),x,p;\lambda ) -( \beta x_{j}-\alpha p_{j})_{j=1}^{J}\Big) \Big|z \Big]\geq 0 \\ 
\text{for all }v\in \mathbb{R}^{J}:\Vert v\Vert =1\text{ and a.e. }z\in \mathbb{S}_{z}
\end{array}
\right\} ,\label{eq:Id_set_expression}
\end{equation}
where $\sigma ^{-1}$ is the inverse of the function in \eqref{eq:sigma_defn} with respect to the first argument.
\end{theorem}

Theorem \ref{thm:Id_set_expression} is the main identification result of the paper. It characterizes the sharp identified set $\Theta_I(P)$ for the structural parameters when outside shares are unobserved but are restricted to lie in $S_0(w)$. The key complication is that, for a given value of $\theta$, the model may be consistent with multiple demand shocks, depending on the admissible value of the outside share. This is captured by the set $\mathcal{U}(w;\theta)$, which collects all demand shocks that are compatible with the observables, the candidate parameter value, and the restriction on the outside share. Characterizing the existence of admissible shocks satisfying the BLP moment restrictions leads to the support-function inequalities in \eqref{eq:Id_set_expression}. The characterization is sharp and also leads naturally to feasible inference procedures, as we show in Section \ref{sec:Inference}.

The characterization is derived using random-set theory, which is well-suited to this setting because the model generates a set of admissible shocks rather than a single shock for each value of $\theta$.  As we explain in Section \ref{sec:NeedRST}, simpler representations based on repeated applications of the standard BLP inversion are generally not sufficient to recover $\Theta_I(P)$.

Theorem \ref{thm:Id_set_expression} relies on two regularity conditions. The first is that the conditional distribution of $(\tilde{s},x,p)$ given $z$ is non-atomic for almost every $z\in\mathbb S_z$; see \citet[page 35]{billingsley:1995}. This condition is satisfied whenever at least one component of $(\tilde{s},x,p)$ is continuously distributed conditional on $z$. It is mild in our setting, since applications typically assume that either $\tilde{s}$ or $p$ is conditionally continuously distributed. The role of this condition is to ensure that the conditional selection expectation of the relevant random set is almost surely convex. This allows us to reexpress the IV exogeneity condition
to be written equivalently as the support-function inequalities in \eqref{eq:Id_set_expression}; see, for example, \citet[Theorem 3.7]{molchanov/molinari:2018}.

The second regularity condition is the integrability of the random set $\mathcal{U}(w;\theta)$ in the sense of selection expectations. That is, $\mathcal{U}(w;\theta)$ must admit at least one integrable selection; see \citet[Definition 3.1]{molchanov:2005} and \citet[page 70]{molchanov/molinari:2018}. This condition ensures that the relevant conditional selection expectation $E[\mathcal{U}(w;\theta)\mid z]$ is well defined. Importantly, it is weaker than integrable boundedness and allows $\mathcal{U}(w;\theta)$ to be unbounded. When $\mathcal{U}(w;\theta)$ is unbounded, its support function may be infinite in some directions. Such directions impose no restrictions on the structural parameter, while finite-valued directions generate the nonredundant restrictions. This distinction is important for our purposes, as illustrated in the running example below.

\begin{runningexample}[Computing objects in Theorem \ref{thm:Id_set_expression} in the plain logit model]
Consider the plain logit case (i.e., $\Lambda =\{\bar{ \lambda}\}$), with $x_{j}=(1,\tilde{x}_{j}')'$ and $\beta =(\beta_{1},\beta_{2})$. For concreteness, we focus on the fully agnostic case, i.e., $S_0(w)=(0, 1)$.  In this case, \eqref{eq:U_defn} becomes
\begin{align*}
\mathcal{U}(w;\theta)~&\overset{(1)}{=}~ \left\{\xi \in \mathbb{R}^{J}: \tilde{s}~=~\left(\frac{\exp (\beta x_{j}-\alpha p_{j}+\xi _{j} )}{\sum_{b=1}^{J}\exp (\beta x_{b}-\alpha p_{b}+\xi _{b} )}\right)_{j=1}^{J} \right\} \\
~&\overset{(2)}{=}~\left\{ \xi \in \mathbb{R}^{J}:
\left\{\begin{array}{c}
\xi _{j}-\xi _{k}=\ln (\tilde{s}_{j}/ \tilde{s}_{k})-\beta _{2}'(\tilde{x}_{j}-\tilde{x}_{k})+\alpha (p_{j}-p_{k})\\
\text{ for all }j,k=1,\dots,J
\end{array}\right\}
\right\} ,
\end{align*}
where (1) holds by $\sigma_0\big((\beta x_j-\alpha p_j+\xi_j)_{j=1}^J,x,p;\lambda\big)\in (0,1)=: S_0(w)$, and (2) by $x_{j}=(1,\tilde{x}_{j}')'$ and $\beta =(\beta _{1},\beta _{2})$. Given one of the demand shocks (say, $\xi _{1}$), the rest of the demand shocks in $\xi$ are determined by the observables and structural parameters.

We now turn to the identified set $\Theta_I(P)$. Since $\lambda\in\Lambda=\{\bar{\lambda}\}$ is point identified by assumption, it remains to characterize the restrictions on $\alpha$, $\beta_1$, and $\beta_2$. By \eqref{eq:sigmaj_inv}, the expression inside the conditional expectation in \eqref{eq:Id_set_expression} is
\begin{equation*}
 \sup_{s_{0}\in S_0(w)}\left(\ln ({(1-s_{0})}/{s_{0}})-\beta _{1} \right)\Big(\sum\nolimits_{j=1}^{J}v_{j}\Big) +v^{\prime }\Big(\ln (\tilde{s}_{j})-\beta _{2}'\tilde{x}_{j}+\alpha p_{j}\Big)_{j=1}^{J}.
\end{equation*}
If $\sum_{j=1}^{J}v_j\neq 0$, the supremum over $s_0\in(0,1)$ is equal to infinity. Hence, the corresponding support-function inequality is automatically satisfied and imposes no restriction on $\theta$. The only directions that restrict $\theta$ are those satisfying $\sum_{j=1}^{J}v_j=0$. For these directions, the common outside-share component vanishes, and \eqref{eq:Id_set_expression} reduces to
\begin{equation}
E \big[ v^{\prime } \big( \ln (\tilde{s}_{j})-\beta_{2}'\tilde{x}_{j}+\alpha p_{j} \big)_{j=1}^{J} ~\big|~ z \big]~\geq~ 0
\label{eq:example_id_set_1}
\end{equation}
for all $v\in\mathbb R^J$ such that $\Vert v\Vert=1$ and $\sum_{j=1}^J v_j=0$, and for a.e.\ $z\in\mathbb S_z$. Since this collection of directions is symmetric, in the sense that $v$ is admissible if and only if $-v$ is admissible, the inequalities in \eqref{eq:example_id_set_1} are equivalent to the equalities
\begin{equation}
E\big[v^{\prime }\big(\ln (\tilde{s}_{j})-\beta_{2}'\tilde{x}_{j}+\alpha p_{j}\big)_{j=1}^{J}~\big|~  z\big]~=~0
\label{eq:example_id_set_2}
\end{equation}
for all $v\in\mathbb R^J$ such that $\Vert v\Vert=1$ and $\sum_{j=1}^J v_j=0$, and for a.e.\ $z\in\mathbb S_z$. Equivalently, these restrictions can be written as the pairwise conditional moment equalities
\begin{equation}
E\left[
\ln (\tilde{s}_{j}/\tilde{s}_{k})
-\beta_{2}'\left( \tilde{x}_{j}-\tilde{x}_{k}\right)
+\alpha \left( p_{j}-p_{k}\right)
\mid z
\right]
=0
\label{eq:example_id_set_3}
\end{equation}
for all $j,k=1,\dots,J$ with $j\neq k$ and a.e.\ $z\in\mathbb S_z$. Therefore, the identified set in \eqref{eq:Id_set_expression} can be equivalently expressed as
\begin{equation}
\Theta _{I}(P)=\left\{ 
\begin{array}{c}
(\alpha ,\beta _{1},\beta _{2},\lambda )\in \Theta =\mathbb{R}\times \mathbb{R}\times \mathbb{R}^{d_{X}-1}\times\{\bar{\lambda}\}: \\ 
E[ \ln (\tilde{s}_{j}/\tilde{s}_{k})-\beta _{2}'( \tilde{x}_{j}-\tilde{x}_{k}) +\alpha ( p_{j}-p_{k}) |z] =0 \\ 
\text{for all }j,k=1,\dots ,J\text{ with }j\neq k\text{ and a.e. }z\in \mathbb{S}_{z}
\end{array}
\right\} .
\end{equation}
Naturally, the conclusions coincide with our earlier ad-hoc derivations in Example \ref{ex:NoId}. The parameter $\lambda$ is point identified by assumption, $\beta _{1}$ is completely unidentified, and $\alpha$ and $\beta_{2}$ are restricted by the pairwise conditional moment equalities above. Provided that $E[ \tilde{x}_{j}-\tilde{x}_{k}|z] $ and $E[ p_{j}-p_{k}|z] $ vary sufficiently for some $j,k$, $\alpha$ and $\beta_{2} $ are point identified.
\end{runningexample}

Theorem \ref{thm:Id_set_expression} immediately delivers identified sets for components or subvectors of the structural parameters. These sets are obtained by projecting $\Theta_I(P)$ onto the corresponding coordinates. For example, the identified set for the price coefficient $\alpha$ is obtained by projecting $\Theta_I(P)$ onto the $\alpha$ coordinate:
\begin{equation*}
\Theta_{I}^\alpha(P)=
\left\{
\alpha\in\mathbb R:
(\alpha,\beta,\lambda)\in\Theta_I(P)
\text{ for some }(\beta,\lambda)
\right\}.
\label{eq:projection}
\end{equation*}
Identified sets for other components or subvectors of $\theta$ are obtained analogously.

\subsection{Can we get the identified set by repeated applications of BLP?}\label{sec:NeedRST}

Given the result in Theorem \ref{thm:Id_set_expression}, a natural question is whether $\Theta_I(P)$ can be recovered by repeatedly applying the standard BLP identification argument, varying the outside share over all admissible values $s_0\in S_0(w)$ for each realization of the data. The next result shows that this is possible in principle, but that the resulting characterization is of limited practical use because it requires searching over outside-share choices as a function of the full vector of observables.

\begin{theorem}[Alternative characterization of the identified set]\label{thm:id_sets1}
Let $\mathbb{S}_{w}$ denote the support of the observed data $w = (\tilde s, x, p, z)$. Under the conditions of Theorem \ref{thm:Id_set_expression},
\begin{align*}
\Theta_I(P) &~=~\left\{ 
\begin{array}{c}
( \alpha ,\beta ,\lambda ) \in \Theta:\exists s^*_{0}:\mathbb{S}_{w}\to (0,1), \text{ with } s^*_0(w) \in S_0(w) \text{ for all }w \in \mathbb{S}_{w},\\ 
E\big[ \sigma ^{-1}( \tilde{s}(1-s^*_{0}(w) ),x,p;\lambda ) -\left( \beta x_{j}-\alpha p_{j}\right) _{j=1}^{J}\big| z\big] ={\bf 0}_{J \times 1}~\text{a.e.}~z\in \mathbb{S}_{z}
\end{array}
\right\} 
\end{align*}
\end{theorem}

Theorem \ref{thm:id_sets1} is intuitive. Since the source of partial identification is that $s_0$ can only be restricted to lie in $S_0(w)$, one can generate $\Theta_I(P)$ by considering all possible values of this outside share for each data realization. This characterization, however, is not well suited for implementation. In a finite sample, $w = (\tilde s, x, p, z)$ can take as many values as there are markets, and the number of markets in a dataset can be large. In fact, if $w$ contains a continuously distributed component (e.g., product prices), then there will typically be as many realizations of $w$ as markets in the data. By contrast, the equivalent representation of $\Theta_I(P)$ in Theorem \ref{thm:Id_set_expression} is more amenable to inference; see Section \ref{sec:Inference}.

Given this difficulty, one may ask whether a simpler characterization is available under additional restrictions. In particular, suppose that the restriction on the outside share depends on $w$ only through the instrument $z$, so that $S_0(w)$ can be written as $S_0(z)$. This includes, for example, the fully agnostic case in which $S_0(w)=(0,1)$. One might then try to replace the full class of functions $s_0^*:\mathbb S_w\to(0,1)$, satisfying $s_0^*(w)\in S_0(z)$, with the simpler class of functions $s_0^*:\mathbb S_z\to(0,1)$, satisfying $s_0^*(z)\in S_0(z)$.\footnote{This nests the common practice in which outside-good shares are assumed to be functions of potentially exogenous observables. The parameters of this function are treated as fixed, and sensitivity analysis is conducted by applying the standard BLP framework to each possible parameter value within a prespecified range. Such an approach further limits attention to a particular class of functions.} The following result shows that this simplification can be too restrictive: it yields a subset of $\Theta_I(P)$ and may even produce an empty set, thereby excluding the true parameter value.

\begin{theorem}[Simpler characterization of the identified set may not work]\label{thm:id_sets2}
Assume the conditions of Theorem \ref{thm:Id_set_expression} and that the outside share depends on $w = (\tilde s, x, p, z)$ only through $z$, and so $S_0(w)=S_0(z)$ for all $w\in \mathbb{S}_w$. Define the following set:
\begin{align*}
\mathcal{H}(P)  &~=~\left\{ 
\begin{array}{c}
( \alpha ,\beta ,\lambda ) \in \Theta:\exists s^*_{0}:\mathbb{S}_{z}\to (0,1), \text{ with } s^*_0(z) \in S_0(z) \text{ for all }z \in \mathbb{S}_{z},\\ 
E\big[ \sigma ^{-1}( \tilde{s}(1-s^*_{0}(z) ),x,p;\lambda ) -\left( \beta x_{j}-\alpha p_{j}\right)_{j=1}^{J}\big| z\big] ={\bf 0}_{J \times 1}~\text{a.e.}~z\in \mathbb{S}_{z}
\end{array}
\right\}.
\end{align*}
Then, $\mathcal{H}(P)\subseteq  \Theta _{I}(P)$ and, in general, $\Theta _{I}(P) \neq \mathcal{H}(P)$. In fact, it is possible to have data generating processes in which $\mathcal{H}(P)$ is an empty set and $\Theta _{I}(P)$ is not an empty set.
\end{theorem}

Theorems \ref{thm:id_sets1} and \ref{thm:id_sets2} underscore that recovering $\Theta_I(P)$ requires searching over a rich class of functions. Any simplification of this space can produce a strict subset of $\Theta_I(P)$, potentially even empty. In what follows, we adopt the characterization of $\Theta_I(P)$ in Theorem \ref{thm:Id_set_expression} (equivalently, Theorem \ref{thm:id_sets1}). Despite its complexity, the representation in \eqref{eq:Id_set_expression} is tractable and amenable to standard inference methods, as shown in Section \ref{sec:Inference}.

\subsection{Identification of other quantities of interest}

In addition to structural parameters, one may be interested in other economically relevant quantities in the BLP model. As shown in Section \ref{sec:standard_BLP}, these objects are functions of the structural parameters $\theta $, the observables $(s,p,x,z)$, and the demand shocks $\xi $. 

For concreteness, we begin our analysis with the own-price elasticities. Recall from Section \ref{sec:standard_BLP} that the own-price elasticity of product $j=1,\dots,J$ is given by
\begin{equation*}
e_{jj} ~=~-\frac{p_{j}}{s_{j}}\int_{(\zeta ,\nu )}(\alpha +\nu )
\left(\begin{array}{c}
\tfrac{\exp (\beta x_{j} -\alpha p_{j}+\xi _{j}+\zeta x_{j} -\nu p_{j})}{1+\sum_{b=1}^{J}\exp (\beta x_{b} -\alpha p_{b}+\xi _{b}+\zeta x_{b} -\nu p_{b})}\times \\
\left( 1-\tfrac{\exp (\beta x_{j} -\alpha p_{j}+\xi _{j}+\zeta x_{j} -\nu p_{j})}{1+\sum_{b=1}^{J}\exp (\beta x_{b} -\alpha p_{b}+\xi _{b}+\zeta x_{b} -\nu p_{b})} \right)
\end{array}\right)
f(\zeta ,\nu ;\lambda )d(\zeta ,\nu ).
\end{equation*}
The right-hand side expression presents three identification challenges when $s_0$ is unobserved. First, the structural parameter $\theta = (\alpha,\beta,\lambda)$ is partially identified, and can only be restricted to its identified set $\Theta_I(P)$ in \eqref{eq:Id_set_expression}. Second, the demand shock $\xi$ is also partially identified, and can only be restricted to $\mathcal{U}(w;\theta)$ in \eqref{eq:U_defn} for each combination of observables and structural parameters $\theta$. These two problems are addressed in Theorem \ref{thm:Id_set_expression}. Third and finally, $s_j = \tilde s_j (1-s_0)$ is unknown. To deal with this issue, we combine \eqref{eq:sharesUncond} and \eqref{eq:sigma_defn} to express the unobserved share $s_j$ as follows:
\begin{align}
s_j ~=~ \int_{(\zeta ,\nu )}\tfrac{\exp \left( \beta x_{j}-\alpha p_{j}+\xi _{j}+ \zeta x_{j} - \nu p_{j} \right) }{1+\sum_{b=1}^{J} \exp \left( \beta x_{b}-\alpha p_{b}+\xi _{b}+\zeta x_{b}  -\nu p_{b} \right) }f( \zeta ,\nu ;\lambda ) d(\zeta ,\nu ).
\label{eq:shares_BLP}
\end{align}
By plugging this in \eqref{eq:elasticity_own}, we obtain the following formula for the own-elasticity:
\begin{equation}
e_{jj}~=~-p_{j}\tfrac{\int_{(\zeta ,\nu )}(\alpha +\nu )
\left(\begin{array}{c}
\tfrac{\exp (\beta x_{j} -\alpha p_{j}+\xi _{j}+\zeta x_{j} -\nu p_{j})}{1+\sum_{b=1}^{J}\exp (\beta x_{b} -\alpha p_{b}+\xi _{b}+\zeta x_{b} -\nu p_{b})}\times \\
\left( 1-\tfrac{\exp (\beta x_{j} -\alpha p_{j}+\xi _{j}+\zeta x_{j} -\nu p_{j})}{1+\sum_{b=1}^{J}\exp (\beta x_{b} -\alpha p_{b}+\xi _{b}+\zeta x_{b} -\nu p_{b})} \right)
\end{array}\right)
f(\zeta ,\nu ;\lambda )d(\zeta ,\nu )}{\int_{(\zeta ,\nu )}\tfrac{ \exp \left( \beta x_{j}-\alpha p_{j}+\xi_{j}+ \zeta x_{j} - \nu p_{j} \right) }{1+\sum_{b=1}^{J} \exp \left( \beta x_{b}-\alpha p_{b}+\xi _{b}+\zeta x_{b}  -\nu p_{b} \right) }f( \zeta ,\nu ;\lambda ) d(\zeta ,\nu )}.
\label{eq:elasticity_own2}
\end{equation}
By combining \eqref{eq:elasticity_own2} with the identified sets in Theorem \ref{thm:Id_set_expression}, we obtain an identified set for the own-price elasticity. This identified set is denoted by $\mathcal{E}_{jj}(\tilde{s},x,p)$ and specified in the next theorem. The result also includes identified sets for the cross-price elasticities, markups, and diversion ratios.

\begin{theorem}[Identified set for other quantities of interest]\label{thm:eq_objects}
Assume the conditions in Theorem \ref{thm:Id_set_expression}. Then, the identified set for the own-price elasticity of product $j=1,\dots, J$ in a market with observables $w= (\tilde{s},x,p,z)$ is given by:
\begin{equation*}
\mathcal{E}_{jj}(w)~=~\left\{
\begin{array}{c}
e_{jj} = -p_{j}\tfrac{\int_{(\zeta ,\nu )}(\alpha +\nu )
\left(\begin{array}{c}
\tfrac{\exp (\beta x_{j} -\alpha p_{j}+\xi _{j}+\zeta x_{j} -\nu p_{j})}{1+\sum_{b=1}^{J}\exp (\beta x_{b} -\alpha p_{b}+\xi _{b}+\zeta x_{b} -\nu p_{b})}\times \\
\left( 1-\tfrac{\exp (\beta x_{j} -\alpha p_{j}+\xi _{j}+\zeta x_{j} -\nu p_{j})}{1+\sum_{b=1}^{J}\exp (\beta x_{b} -\alpha p_{b}+\xi _{b}+\zeta x_{b} -\nu p_{b})} \right)
\end{array}\right)
f(\zeta ,\nu ;\lambda )d(\zeta ,\nu )}{\int_{(\zeta ,\nu )}\tfrac{ \exp \left( \beta x_{j}-\alpha p_{j}+\xi_{j}+ \zeta x_{j} -\nu p_{j} \right) }{1+\sum_{b=1}^{J} \exp \left( \beta x_{b}-\alpha p_{b}+\xi _{b}+\zeta x_{b}  -\nu p_{b} \right) }f( \zeta ,\nu ;\lambda ) d(\zeta ,\nu )}:\\
 \xi  \in \mathcal{U}(w;\theta )~~\text{and}~~\theta \in \Theta _{I}(P)
\end{array}
\right\} .
\end{equation*}
In turn, the identified set for the cross-price elasticity of product $j=1,\dots, J$ with respect to the price of product $k\neq j$ in a market with observables $w$ is given by:
\begin{equation*}
\mathcal{E}_{jk}(w)~=~\left\{
\begin{array}{c}
e_{jk} = p_{k}\tfrac{\int_{(\zeta ,\nu )}(\alpha +\nu )
\left(\begin{array}{c}
\tfrac{\exp (\beta x_{j} -\alpha p_{j}+\xi _{j}+\zeta x_{j} -\nu p_{j})}{1+\sum_{b=1}^{J}\exp (\beta x_{b} -\alpha p_{b}+\xi _{b}+\zeta x_{b} -\nu p_{b})}\times \\
\tfrac{\exp (\beta x_{k} -\alpha p_{k}+\xi _{k}+\zeta x_{k} -\nu p_{k})}{1+\sum_{b=1}^{J}\exp (\beta x_{b} -\alpha p_{b}+\xi _{b}+\zeta x_{b} -\nu p_{b})}
\end{array}\right)
f(\zeta ,\nu ;\lambda )d(\zeta ,\nu )}{\int_{(\zeta ,\nu )}\tfrac{ \exp \left( \beta x_{j}-\alpha p_{j}+\xi_{j}+ \zeta x_{j} -\nu p_{j} \right) }{1+\sum_{b=1}^{J} \exp \left( \beta x_{b}-\alpha p_{b}+\xi _{b}+\zeta x_{b}  -\nu p_{b} \right) }f( \zeta ,\nu ;\lambda ) d(\zeta ,\nu )}:\\
 \xi  \in \mathcal{U}(w;\theta )~~\text{and}~~\theta \in \Theta _{I}(P)
\end{array}
\right\}.
\end{equation*}
Also, the identified set for the markup of product $j=1,\dots,J$ in a market with observables $w$ is given by:
\begin{equation*}
\mathcal{M}_{j}(w)~=~\left\{
\begin{array}{c}
M_{j} =  \tfrac{
\int_{(\zeta ,\nu )}\tfrac{ \exp \left( \beta x_{j}-\alpha p_{j}+\xi_{j}+ \zeta x_{j} -\nu p_{j} \right) }{1+\sum_{b=1}^{J} \exp \left( \beta x_{b}-\alpha p_{b}+\xi _{b}+\zeta x_{b}  -\nu p_{b} \right) }f( \zeta ,\nu ;\lambda ) d(\zeta ,\nu )
}{
\int_{(\zeta ,\nu )}(\alpha +\nu )
\left(\begin{array}{c}
\tfrac{\exp (\beta x_{j} -\alpha p_{j}+\xi _{j}+\zeta x_{j} -\nu p_{j})}{1+\sum_{b=1}^{J}\exp (\beta x_{b} -\alpha p_{b}+\xi _{b}+\zeta x_{b} -\nu p_{b})}\times \\
\left( 1-\tfrac{\exp (\beta x_{j} -\alpha p_{j}+\xi _{j}+\zeta x_{j} -\nu p_{j})}{1+\sum_{b=1}^{J}\exp (\beta x_{b} -\alpha p_{b}+\xi _{b}+\zeta x_{b} -\nu p_{b})} \right)
\end{array}\right)
f(\zeta ,\nu ;\lambda )d(\zeta ,\nu )
}:\\
 \xi  \in \mathcal{U}(w;\theta )~~\text{and}~~\theta \in \Theta _{I}(P)
\end{array}
\right\}.
\end{equation*}
Finally, the identified set for the diversion ratio of product $j = 1,\dots, J$ with respect to product $k \ne j$ in a market with observables $w$ is given by:
\begin{equation*}
\mathcal{D}_{jk}(w)~=~\left\{
\begin{array}{c}
D_{jk} = \tfrac{\int_{(\zeta ,\nu )}(\alpha +\nu )
\left(\begin{array}{c}
\tfrac{\exp (\beta x_{j} -\alpha p_{j}+\xi _{j}+\zeta x_{j} -\nu p_{j})}{1+\sum_{b=1}^{J}\exp (\beta x_{b} -\alpha p_{b}+\xi _{b}+\zeta x_{b} -\nu p_{b})}\times \\
 \tfrac{\exp (\beta x_{k} -\alpha p_{k}+\xi _{k}+x_{k}\zeta -\nu p_{k})}{1+\sum_{b=1}^{J}\exp (\beta x_{b} -\alpha p_{b}+\xi _{b}+\zeta x_{b} -\nu p_{b})} 
\end{array}\right)
f(\zeta ,\nu ;\lambda )d(\zeta ,\nu )}{\int_{(\zeta ,\nu )}(\alpha + \nu) 
\left(\begin{array}{c}
\tfrac{\exp (\beta x_{k} -\alpha p_{k}+\xi _{k}+x_{k}\zeta -\nu p_{k})}{1+\sum_{b=1}^{J}\exp (\beta x_{b} -\alpha p_{b}+\xi _{b}+\zeta x_{b} -\nu p_{b})}\times \\
\left( 1-\tfrac{\exp (\beta x_{k} -\alpha p_{k}+\xi _{k}+x_{k}\zeta -\nu p_{k})}{1+\sum_{b=1}^{J}\exp (\beta x_{b} -\alpha p_{b}+\xi _{b}+\zeta x_{b} -\nu p_{b})} \right)
\end{array}\right)
f(\zeta,\nu;\lambda)d(\zeta,\nu)}:\\
 \xi  \in \mathcal{U}(w;\theta)\text{ and }\theta \in \Theta _{I}(P)
\end{array}
\right\}.
\end{equation*}
\end{theorem}

We now illustrate these in the context of the plain logit model.

\begin{runningexample}[Identified sets for quantities of interest in the plain logit model]
Consider the plain logit case (i.e., $\Lambda =\{\bar{ \lambda}\}$), with $x_{j}=(1,\tilde{x}_{j}')'$ and $\beta =(\beta _{1},\beta _{2})$. For each $j,k=1,\dots,J$ with $j\neq k$, an observables $w$,
\begin{align*}
\mathcal{E}_{jj}(w) ~&=~\Big\{e_{jj}=-\alpha p_j \big(1-\tfrac{\exp(\beta_1+\beta_2\tilde x_j-\alpha p_j+\xi_j)} {1+\sum_{b=1}^J \exp(\beta_1+\beta_2\tilde x_b-\alpha p_b+\xi_b)} \big): \xi\in\mathcal{U}(w;\theta)~\text{and}~\theta\in\Theta_I(P)\Big\},\\
\mathcal{E}_{jk}(w)~&=~
\Big\{e_{jk}=\alpha p_k \tfrac{\exp(\beta_1+\beta_2\tilde x_k-\alpha p_k+\xi_k)} {1+\sum_{b=1}^J \exp(\beta_1+\beta_2\tilde x_b-\alpha p_b+\xi_b)}: \xi\in\mathcal{U}(w;\theta)~\text{and}~\theta\in\Theta_I(P)\Big\},\\
\mathcal{M}_{j}(w)~&=~ \Big\{m_{j}=\frac{1}{\alpha}\big(1- \tfrac{\exp(\beta_1+\beta_2\tilde x_j-\alpha p_j+\xi_j)}{1+\sum_{b=1}^J \exp(\beta_1+\beta_2\tilde x_b-\alpha p_b+\xi_b)} \big): \xi\in\mathcal{U}(w;\theta)~\text{and}~\theta\in\Theta_I(P)\Big\},\\
\mathcal{D}_{jk}(w)~&=~
\Bigg\{d_{jk}=\tfrac{\tfrac{\exp(\beta_1+\beta_2\tilde x_j-\alpha p_j+\xi_j)} {1+\sum_{b=1}^J \exp(\beta_1+\beta_2\tilde x_b-\alpha p_b+\xi_b)} }{\big(1-\tfrac{\exp(\beta_1+\beta_2\tilde x_k-\alpha p_k+\xi_k)} {1+\sum_{a=1}^J \exp(\beta_1+\beta_2\tilde x_a-\alpha p_a+\xi_a)}\big)}: \xi\in\mathcal{U}(w;\theta)~\text{and}~\theta\in\Theta_I(P)\Bigg\}.
\end{align*}

By our earlier derivations, $\beta_1$ is completely unidentified and may take any value in $\mathbb{R}$. Holding $(\alpha,\beta_2,\xi)$ fixed with $\xi\in\mathcal{U}(w;\theta)$, we obtain
\begin{align*}
\lim_{\beta_1\to-\infty}\frac{\exp(\beta_1+\beta_2\tilde x_j-\alpha p_j+\xi_j)}{1+\sum_{b=1}^J \exp(\beta_1+\beta_2\tilde x_b-\alpha p_b+\xi_b)}&=0,\\
\lim_{\beta_1\to\infty}\frac{\exp(\beta_1+\beta_2\tilde x_j-\alpha p_j+\xi_j)}{1+\sum_{b=1}^J \exp(\beta_1+\beta_2\tilde x_b-\alpha p_b+\xi_b)}&=\frac{\exp(\beta_2\tilde x_j-\alpha p_j+\xi_j)} {\sum_{b=1}^J \exp(\beta_2\tilde x_b-\alpha p_b+\xi_b)}~\overset{(1)}{=}~\tilde s_j,
\end{align*}
where (1) holds by $\xi\in\mathcal{U}(w;\theta)$. As $\beta_1$ ranges over $\mathbb{R}$, this expression ranges over $(0,\tilde s_j)$. By continuity of the mappings involved, we obtain the following identified sets:
\begin{align*}
\mathcal{E}_{jj}(w)~&=~\Big\{e_{jj} \in (-\alpha p_j ,-\alpha p_j (1-\tilde s_j))~:~ (\alpha,\beta_1,\beta_2,\bar\lambda)\in\Theta_I(P)\Big\}\\
\mathcal{E}_{jk}(w)~&=~\Big\{e_{jk} \in (0, \alpha p_k \tilde s_k)~:~ (\alpha,\beta_1,\beta_2,\bar\lambda)\in\Theta_I(P)\Big\}\\
\mathcal{M}_{j}(w)~&=~\Big\{m_{j} \in ({1}/{\alpha },{1}/{(\alpha (1-\tilde s_j))})~:~ (\alpha,\beta_1,\beta_2,\bar\lambda)\in\Theta_I(P)\Big\},\\
\mathcal{D}_{jk}(w)~&=~\Big\{d_{jk} \in (0, {\tilde s_j}/{(1- \tilde s_k)})~:~ (\alpha,\beta_1,\beta_2,\bar\lambda)\in\Theta_I(P)\Big\}.
\end{align*}
Provided that $E[\tilde{x}_{j}-\tilde{x}_{k}\mid z]$ and $E[p_{j}-p_{k}\mid z]$ vary sufficiently, $\alpha$ was shown to be point identified. If so, the first three identified sets simplify further, with $\alpha$ fixed at its uniquely identified value.
\end{runningexample}


\subsection{Numerical illustration}\label{sec:numericalWork}

This section presents a numerical illustration of the identified sets for the structural parameters and other quantities of interest in the BLP model when the outside-good share is partially identified.

The data-generating process is a simplified version of the simulation design in \citet[Section 8.3]{gandhi/lu/shi:2023} (henceforth GLS). In particular, we consider markets with $J=2$ inside goods and an outside good. Assume that there is a vector of continuous instruments $z = (z_{j})_{j=1}^{J}$, with $z_{j} \sim U([0,2])$. The instrument shifts demand for products and affects the prices of the inside goods. Following the setup in GLS, prices of the inside goods are generated according to
\begin{align*}
p_{j} ~=~ z_{j} + b \exp(\xi_{j})~~~~\text{for}~j=1,\dots,J,
\end{align*}
where we set $b=1$. Demand shocks are generated from:
\begin{equation*}
\xi_{j} ~=~ \mathbf{1}\left[\beta x_{j}   \ge med(\beta x_{j} )\right]\xi_{j}^{\prime} + \mathbf{1}\left[ \beta x_{j}  < med(\beta x_{j} ) \right] \xi_{j}^{\prime\prime}~~~~\text{for}~j=1,\dots,J,
\end{equation*} 
with $\xi_{j}^{\prime}$ and $\xi_{j}^{\prime\prime}$ i.i.d.\ and $N(0, 0.5^{2})$, and $x_{j}$ is defined below.

In this model, prices are endogenous due to 
\begin{equation}
cov[\xi_{j}, p_{j}] ~=~ cov[\xi_{j}, z_{j} + b\exp(\xi_{j})] ~=~ bcov[\xi_{j}, \exp(\xi_{j})]~\neq~0.
\end{equation}
In turn, the instrument is valid since $\xi_{j} \perp z$.

Demand is generated from the BLP model described in Section \ref{sec:standard_BLP}, with the random-coefficients specification specialized to allow heterogeneity only in price sensitivity. That is, the market shares for inside goods $j=1,2$ and the outside option are given by
\begin{align*}
s_j&~=~\int\frac{\exp(\beta x_{j} - \alpha p_j + \xi_j  - \nu p_{j})}{1 + \sum_{b=1}^{J}\exp(\beta x_{b} - \alpha p_b + \xi_b - \nu p_{b})
}f(\zeta,\nu;\lambda)d(\zeta,\nu),\\
s_0&~=~\int\frac{1}{1 + \sum_{b=1}^{J}\exp(\beta x_{b} - \alpha p_b + \xi_b  - \nu p_{b})}f(\zeta,\nu;\lambda)d(\zeta,\nu),
\end{align*}
where $\nu \sim N(0,\lambda^{2})$. The observed product characteristics $x_{j}$ consist of a constant, a continuous variable, $\tilde{x}_{j} \sim N(-0.5,1)$, and product fixed effects. We set the coefficient of $\tilde{x}_{j}$ to $\beta_{2}=1$, while $(\alpha,\beta_{1},\lambda)$ and product fixed effects are set to point estimates obtained by applying the standard BLP specification with observed $s_0$ to the dataset used in GLS's application after adapting it to our setting. This dataset is later used in the empirical application in Section \ref{sec:Application}. In particular, because handling zero market shares is not the focus of this paper, we adapt the data to our setting by trimming observations with zero shares.

To keep the numerical exercise tractable, we restrict our attention to three parameters:
\begin{equation*}
\theta ~=~ (\alpha,\beta_{1},\lambda) ~\in~ \Theta,
\end{equation*}
where $\alpha$ is the price coefficient, $\beta_{1}$ is the constant term, and $\lambda$ governs heterogeneity in price sensitivity. The remaining parameters of the model are treated as known and fixed at their true values. For computational purposes, we restrict attention to the compact parameter space
\begin{equation*}
\Theta=[0,5]\times[-8,4]\times[0,1].
\end{equation*}
The true parameter values used in the simulation are
\begin{equation*}
\alpha=1,~~~\beta_{1}=-4,~~~\lambda=0.05.
\end{equation*}
As explained earlier, these values are chosen to be broadly consistent with point estimates obtained from a standard BLP specification with the outside share $s_0$ observed.

To compute the identified set, we generate a large simulated dataset with $n = 50{,}000$ markets, thereby allowing us to abstract away from random sampling error. As in GLS, the distribution of the true outside share $s_0$ is concentrated at very high values: the minimum value is $0.8061$, the $25$th quantile is $0.9902$, the median is $0.9947$, the $75$th quantile is $0.9971$, and the maximum value is $0.9999$. These large values will motivate our specification of $S_0(w)$ below.

To implement the identification method with continuous instruments, we discretize $z_{1}$ and $z_{2}$ into $d_{z} = 30$ quantile-based bins. The identified set $\Theta_{I}(P)$ is the intersection of the sets obtained based on $z_{1}$ and $z_{2}$, respectively. 

We study identification under two information structures for the outside-good share:
\begin{enumerate}
\item \emph{Partial information on outside shares:} $S_0(w)=[c,1)$, where $c\geq 0$ is a constant. This represents a common case in applications where the researcher knows that the outside-good share is large, but does not know its exact value. In the numerical exercise below, we use $c=0.8$, a choice guided by the simulated data.\footnote{In Section \ref{sec:Numerical_NoInfo} in the appendix, we report results for the case with no information on $s_{0}$, i.e., $S_0(w) = (0,1)$.}
\item \emph{Small amount of missing information on outside shares:} $S_0(w) = [s_0-\varepsilon,\, s_0+\varepsilon]\cap(0,1)$, where $s_0$ denotes the true outside share and $\varepsilon>0$ is a small constant. This specification examines the sensitivity of identification to small departures from point knowledge of the outside-good share. In the numerical exercise, we set $\varepsilon=0.05$.
\end{enumerate}
We divide the remainder of the section into these two exercises.

\subsubsection{Partial information on outside shares}\label{sec:Numerical1}

We now report the identified sets for the partial-information case with $S_0(w) = [0.8,1)$. Figure \ref{fig:idset_partial} depicts the identified set $\Theta_{I}(P)$. This is a three-dimensional object that we visualize in the $(\alpha,\beta_1)$ plane for different values of $\lambda$. 

We compute this set numerically using \eqref{eq:Id_set_expression}. The moment inequalities are evaluated on a fine grid of values of $\theta=(\alpha,\beta_1,\lambda)$, with expectations replaced by sample analogs from the simulated dataset of $n=50{,}000$ markets. To approximate the requirement that the inequalities hold for all $v\in\mathbb R^2$ with $\Vert v\Vert=1$, we use $360$ directions of the form $v=(\cos a,\sin a)$, with $a$ evenly spaced over $[0,2\pi)$. 

For the conditioning information, we use the instruments $z_1$ and $z_2$. As already mentioned, we discretize the support of each instrument into $d_z=30$ quantile-based bins and compute the conditional expectations within those bins. The conditioning restrictions are imposed separately on $z_{1}$ and $z_{2}$. For each instrument, the identified set is the set of parameter values satisfying the resulting bin-level inequalities. We then report the intersection of the identified sets obtained from $z_1$ and $z_2$.

Finally, to account for numerical approximation error from the finite grids and the fixed-point inversion, we treat a grid point as belonging to the numerical approximation of the identified set whenever the sample analogs of the moment inequalities are no smaller than the negative of a tolerance parameter $-\tau$. In our numerical exercise, we set this tolerance parameter to $\tau = 0.01$. This tolerance is intentionally conservative: increasing $\tau$ enlarges the numerical approximation of the identified set, making us more likely to include rather than exclude parameter values.\footnote{Since $n=50{,}000$, $\tau=0.01$ implies $\sqrt{n}\tau \approx 2.24$, so the tolerance is about $2.24$ times the scale $1/\sqrt{n}$.}

If we project this identified set $\Theta_{I}(P)$ on each parameter coordinate, we obtain the parameter-specific identified sets:
\begin{align*}
\Theta_I^\alpha(P)~=~[0.8,2.9],~~~
\Theta_I^{\beta_{1}}(P)~=~[-5,0],~~~\text{and}~~~
\Theta_I^\lambda(P)~=~[0,1].
\end{align*}
Naturally, the true parameter values $(\alpha,\beta_{1},\lambda) = (1,-4,0.05)$ lie in the corresponding marginal projections of the identified set. The round endpoints reflect the finite grid used in the numerical approximation, rather than exact analytical boundaries. Although these sets are large, they remain informative about the signs and plausible magnitudes of $\alpha$ and $\beta_1$.

\begin{sidewaysfigure}
    \centering
    \includegraphics[width=0.95\textheight]{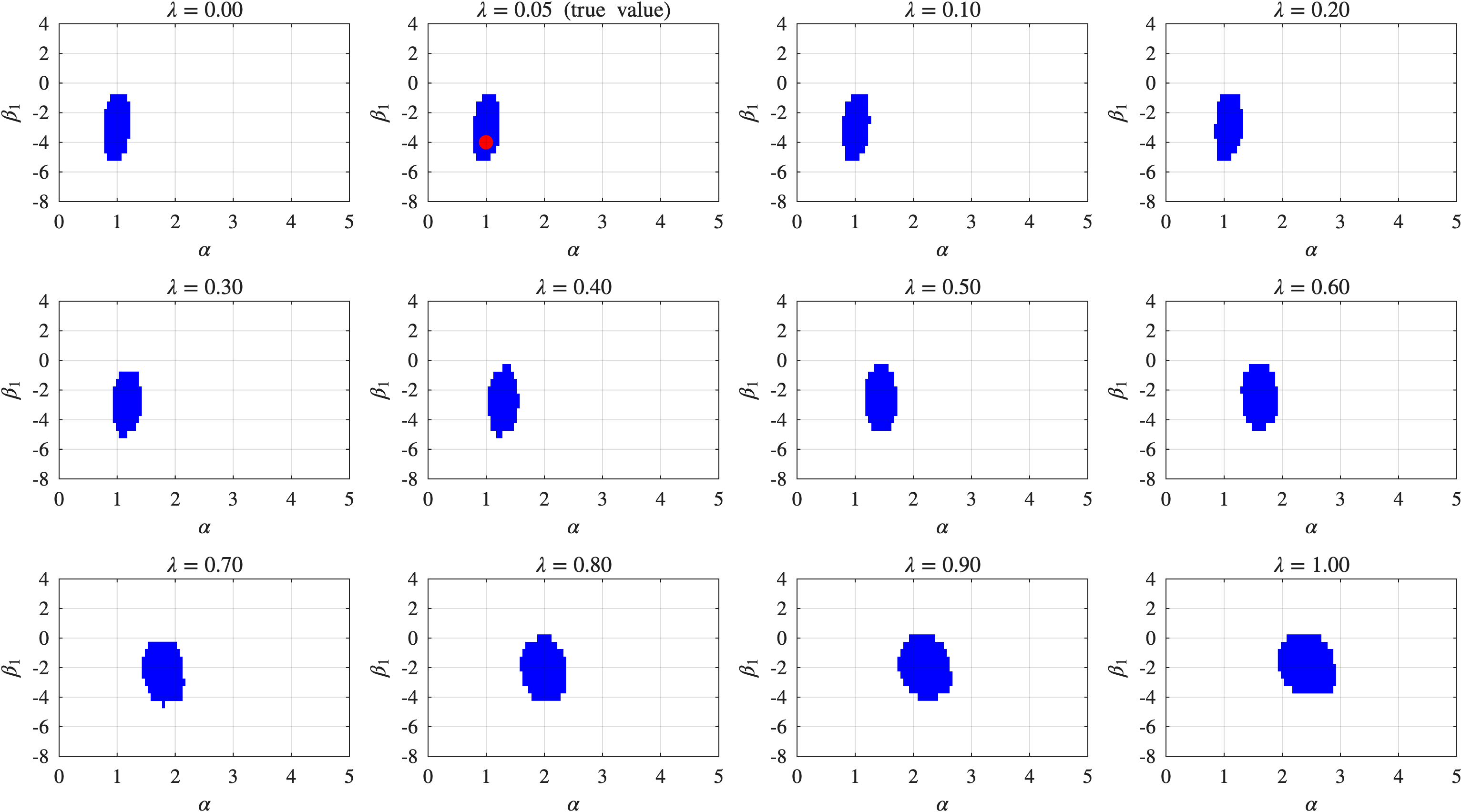}
        \caption{Identified set $\Theta _{I}(P)$ in a DGP with $(\alpha,\beta_{1},\lambda) = (1,-4,0.05)$ and partial information on $s_0$ with $S_0(w)=[0.8,1)$, computed with tolerance parameter $\tau=0.01$.}
    \label{fig:idset_partial}
\end{sidewaysfigure}

Next, we consider demand responses and equilibrium objects in a market with the observables $w$ set at values in the first market in our dataset.\footnote{Following the practice in parts of the IO literature, we report elasticities for a particular market (\cite{Moon:2018aa} and \cite{Aguiar:2025aa}). Other quantities, such as the average or median elasticities across all markets, could also be considered.} In this design, the inside shares of the first market are $\tilde{s} = (0.2787,0.7213)$, and the prices are $p = (1.7523,1.5934)$. 


Using these results, we can compute the identified sets for each elasticity in the first market. 
\begin{align*}
    &\mathcal{E}_{1,1}(w) = [ -2.8753, -0.8275],~
    &&\mathcal{E}_{2,2}(w) = [ -2.3022, -0.8934],\\
    &\mathcal{E}_{1,2}(w) = [ -0.0004, 0.7071],~
    &&\mathcal{E}_{2,1}(w) = [ -0.0002, 0.3050].
\end{align*}
The true elasticity values are $e_{1,1} = -1.7367$, $e_{2,2} = -1.5755$, $e_{1,2} = 0.007$, and $e_{2,1} = 0.003$,  all of which lie within their corresponding identified sets. 
The identified sets for markup and diversion ratios are as follows:
\begin{align*}
    &\mathcal{M}_{1}(w) = [ 0.6095, 2.1176],~
    &&\mathcal{M}_{2}(w) = [ 0.6633, 1.7835],\\
    &\mathcal{D}_{1,2}(w) = [ -0.0002, 0.1172],~
    &&\mathcal{D}_{2,1}(w) = [ -0.0005, 0.2705].
\end{align*}
The true markup values are $M_{1} = 1.009$ and $M_{2} = 1.0113$, and the true diversion ratios are $D_{1,2} = 0.0017$ and $D_{2,1} = 0.0044$
Overall, although the identified sets are wide, they remain informative: the markups are positive, and the sets for cross-price elasticities and diversion ratios are nearly nonnegative, with lower endpoints very close to zero.


\subsubsection{Small amount of missing information on the outside shares}

We now consider the case in which there is only a small amount of missing information about the share of the outside goods. Specifically, we set $S_0(w) = [s_0-0.05,s_0+0.05]\cap(0,1)$. 

Figure \ref{fig:idset_small} depicts the resulting identified set $\Theta_I(P)$. The identified set is considerably smaller than the one obtained in the partial-information case with $S_0(w)=[0.8,1)$, reflecting the additional information about the outside-good share. If we project this identified set $\Theta_{I}(P)$ on each parameter coordinate, we obtain the parameter-specific identified sets:
\begin{align*}
\Theta_I^\alpha(P)~=~[0.8,2.85],~~
\Theta_I^{\beta_{1}}(P)~=~[-5,-1],~~\text{and}~~
\Theta_I^\lambda(P)~=~[0,1].
\end{align*}
Naturally, the true parameter values $(\alpha,\beta_{1},\lambda) = (1,-4,0.05)$ lie in the corresponding marginal projections of the identified set. Compared with the partial-information case $S_0(w)=[0.8,1)$, the additional information about $s_0$ tightens the identified sets for $\alpha$ and $\beta_1$, although the identified set still does not rule out any value of $\lambda \in [0,1]$.

Using the observed values of $w$ in the first market, we next compute identified sets for demand responses and equilibrium objects. Consistent with the parameter-identification results, the additional information about $s_0$ yields substantially tighter identified sets for these objects. For the own- and cross-price elasticities, we obtain
\begin{align*}
&\mathcal{E}_{1,1}(w) = [ -2.5312, -1.3797],~
&&\mathcal{E}_{2,2}(w) = [ -2.3228, -1.2233],\\
&\mathcal{E}_{1,2}(w) = [ 0.0009, 0.2328],~
&&\mathcal{E}_{2,1}(w) = [ 0.004, 0.1003].
\end{align*}
The identified sets for markups and diversion ratios are
\begin{align*}
&\mathcal{M}_{1}(w) = [ 0.6922, 1.2701],~
&&\mathcal{M}_{2}(w) = [ 0.6860, 1.3026],\\
&\mathcal{D}_{1,2}(w) = [ 0.0003, 0.0393],~
&&\mathcal{D}_{2,1}(w) = [ 0.0007, 0.1011].
\end{align*}
All of these sets contain the corresponding true values. Moreover, relative to the partial-information case $S_0(w)=[0.8,1)$, the identified sets are substantially tighter. In particular, the identified sets for the cross-price elasticities and diversion ratios are now strictly positive, showing that even a small amount of information about the outside-good share can sharpen identification of economically relevant objects.

\begin{sidewaysfigure}
    \centering
    \includegraphics[width=\textheight]{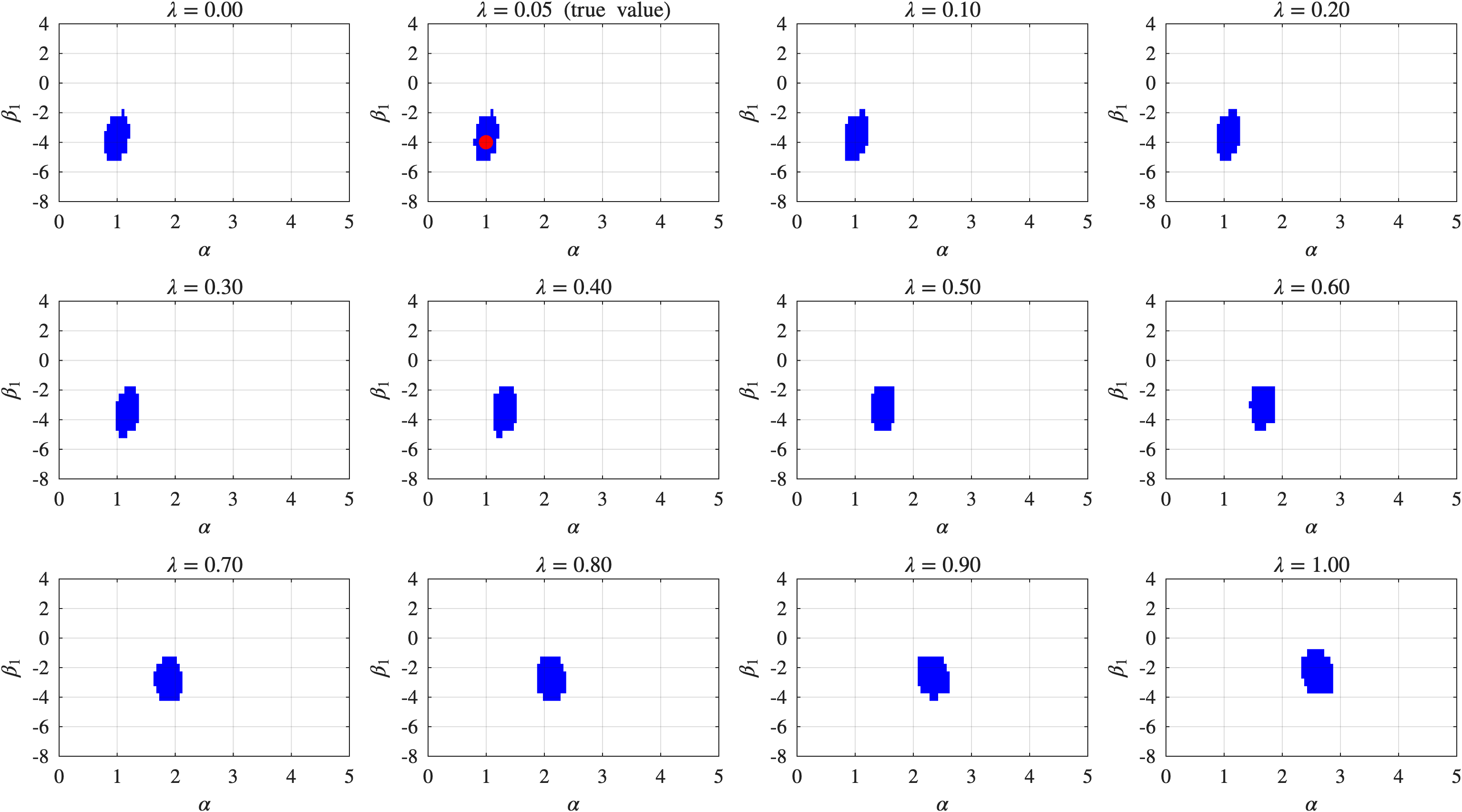}
        \caption{Identified set $\Theta _{I}(P)$ in a DGP with $(\alpha,\beta_{1},\lambda) = (1,-4,0.05)$ and a small amount of missing information on $s_0$, given by $S_0(w)=[s_0-0.05,s_0+0.05]\cap(0,1)$, computed with tolerance parameter $\tau=0.01$.}
    \label{fig:idset_small}
\end{sidewaysfigure}

\section{Inference without knowledge of outside-good shares}\label{sec:Inference}

This section describes how to conduct inference for the BLP model when the outside-good share $s_0$ is unobserved. We conduct inference on both the structural parameters of the BLP model and other economically relevant quantities it implies. Our approach builds directly on ideas from \cite{andrews/shi:2013} and \cite{chernozhukov/chetverikov/kato:2019}.

To explain our inference strategy, we note that the identified set $\Theta_I(P)$ in Theorem \ref{thm:Id_set_expression} is defined as an infinite collection of conditional moment inequalities. To see this, recall $w=( \tilde{s},x,p,z) $, denote $\mathcal V=\{v\in\mathbb R^J:\Vert v\Vert=1\}$, and define
\begin{equation*}
    m( w,v,\theta ) ~=~-\sup_{s_{0}\in S_0(w)}v^{\prime }(\sigma ^{-1}(\tilde{s}(1-s_{0}),x,p;\lambda )-(\beta x_{j}-\alpha p_{j})_{j=1}^{J}). 
\end{equation*}
The support-function inequalities in Theorem \ref{thm:Id_set_expression} allow for the possibility that the supremum is infinite. If, for a given direction $v$, the supremum in the definition of $m(w,v,\theta)$ is equal to infinity with positive probability, then $m(w,v,\theta)$ is equal to minus infinity with positive probability, and the corresponding moment inequality is automatically satisfied. Such directions impose no restriction on $\theta$. They may appear in the population characterization, but they are nonbinding and are not included in the finite collection of moment inequalities used for inference.

With this notation, it follows that $\Theta_I(P)$ can be equivalently expressed as a collection of parameters satisfying conditional moment inequalities:
\begin{equation*}
\Theta _{I}(P)~=~\{ \theta \in \Theta :E[m( w,v,\theta ) |z]\leq 0\text{ for all }v\in \mathcal{V}\text{ and a.e. }z\in \mathbb{S}_{z}\}. 
\end{equation*}

Our next step is to convert the conditional moment inequalities in $\Theta_I(P)$ into a collection of unconditional moment inequalities, following the approach of \cite{andrews/shi:2013}. To this end, let $\mathcal G$ denote a collection of nonnegative functions of the instrument $z$, referred to as {\it instrument functions}. We then define the set:
\begin{equation}
\Theta _{I}(P,\mathcal{G})~=~\{ \theta \in \Theta :E[m( w,v,\theta ) g(z)]\leq 0\text{ for all }v\in \mathcal{V}\text{ and }g\in \mathcal{G}\} .
\label{eq:IS_AS2013}
\end{equation}

Since each $g\in\mathcal G$ is nonnegative, it follows immediately that $\Theta_I(P)\subseteq \Theta_I(P,\mathcal G)$. Thus, $\Theta_I(P,\mathcal G)$ constitutes an {\it outer identified set} for the parameter of interest. However, \cite{andrews/shi:2013} show that for suitable choices of $\mathcal G$, this outer set coincides with the identified set, that is, $\Theta_I(P,\mathcal G)=\Theta_I(P)$; see \citet[Supplemental Appendix B]{andrews/shi:2013}. Examples of such choices include collections of indicator functions over hypercubes or hyperrectangles in $\mathbb{S}_z$. See Example \ref{ex:hypercube} for a description of the definition of the hypercubes. 
Motivated by these results, we choose to conduct inference on $\Theta_I(P,\mathcal G)$. By construction, this set provides an outer identified set for the parameter of interest, and coincides with $\Theta_I(P)$ under appropriate choices of $\mathcal{G}$.

\begin{example}[Hypercubes in \cite{andrews/shi:2013}]\label{ex:hypercube}
The class of instrument functions $\mathcal G_{\mathrm{cube}}$ based on hypercubes is constructed as follows. First, we transform the instrument vector $z=(z_1,\dots,z_{d_z})$ into the unit $d_z$-dimensional hypercube by using a componentwise standard normal CDF:
\begin{equation*}
\tilde z ~=~ (\Phi(z_1),\dots,\Phi(z_{d_z})) ~\in~ [0,1]^{d_z}.
\end{equation*}
Second, we partition $[0,1]^{d_z}$ into regular hypercubes. For a given $r\in\mathbb N$, define
\begin{equation*}
C_{a,r} ~=~ \prod_{u=1}^{d_z}\Big(\frac{a_u-1}{2r},\frac{a_u}{2r}\Big], \qquad a=(a_1,a_2,\dots,a_{d_z}),
\end{equation*}
where $a_u\in\{1,2,\dots,2r\}$ for each $u=1,\dots,d_z$. Each $C_{a,r}$ is a hypercube in $[0,1]^{d_z}$ with side length $(2r)^{-1}$. For any tuning parameter $r_0\in\mathbb N$, the collection of hypercubes is defined by
\begin{equation*}
\mathcal C_{\mathrm{cube}} ~=~ \{ C_{a,r} : a_u\in\{1,2,\dots,2r\},\ u=1,2,\dots,d_z,\ r=r_0,r_0+1,\dots \}.
\end{equation*}
Finally, we define the corresponding class of instrument functions as
\begin{equation*}
\mathcal G_{\mathrm{cube}} ~=~ \{ g(z)=I\{\tilde z\in C\} : C\in\mathcal C_{\mathrm{cube}} \}.
\end{equation*}
Under the regularity conditions described in \cite{andrews/shi:2013}, the collection of hypercubes $\mathcal C_{\mathrm{cube}}$ generates the Borel $\sigma$-algebra on $[0,1]^{d_z}$, and \citet{andrews/shi:2013} shows that this choice of $\mathcal G_{\mathrm{cube}}$ is sufficiently rich so that $\Theta_I(P,\mathcal G_{\mathrm{cube}})=\Theta_I(P)$.
\hfill 
\end{example}

The identified set $\Theta _{I}(P,\mathcal{G})$ in \eqref{eq:IS_AS2013} is defined by a collection of moment inequalities indexed by $(v,g)\in \mathcal{V}\times \mathcal{G}$. Since $\mathcal{V}$ is uncountable and $\mathcal{G}$ typically contains countably infinitely many elements, $\Theta _{I}(P,\mathcal{G})$ is effectively characterized by an infinite number of moment inequalities.

To conduct inference, we rely on the methods for many moment inequalities developed in \cite{chernozhukov/chetverikov/kato:2019}. Specifically, we consider a countable subset $\{(v_{j},g_{j})\}_{j=1}^{p_{n}}\subseteq \mathcal{V}\times \mathcal{G}$, where the number of moment inequalities $p_{n}$ can grow with the sample size $n$, and conduct inference on the approximating identified set
\begin{equation}
\tilde{\Theta}_{I}(P,p_{n})~=~\left\{\theta \in \Theta :E\left[m(w,v_{j},\theta)g_{j}(Z)\right]\leq 0\ \text{for all } j=1,\dots,p_{n}\right\}.
\label{eq:IS_for_CCK}
\end{equation}
For any $n\in \mathbb{N}$, $\Theta _{I}(P)\subseteq \Theta _{I}(P,\mathcal{G})\subseteq \tilde{\Theta }_{I}(P,p_{n})$. Hence, $\tilde{\Theta}_{I}(P,p_{n})$ constitutes an outer identified set for the parameter of interest. Because $p_n$ is allowed to grow with $n$, the resulting confidence sets will remain asymptotically valid in settings with many moment inequalities, including cases in which $p_n$ is large relative to $n$.

The methodology proposed by \cite{chernozhukov/chetverikov/kato:2019} is directly applicable to $\tilde{\Theta}_{I}(P,p_{n})$, as it delivers asymptotically valid inference while allowing $p_{n}$ to grow with the sample size $n$ at relatively fast rates. For each $\theta \in \Theta$, the test statistic is given by
\begin{equation*}
    T_{n}(\theta)~=~\max_{j=1,\dots,p_{n}}\frac{\sqrt{n}\,\hat{\mu}_{j}(\theta)}{\hat{\sigma}_{j}(\theta)},
\end{equation*}
where, for each $j=1,\dots,p_{n}$,
\begin{align*}
\hat{\mu}_{j}(\theta) &~=~\frac{1}{n}\sum_{i=1}^{n} m(w_{i},v_{j},\theta) g_{j}(Z_{i}), \\
\hat{\sigma}_{j}^{2}(\theta) &~=~\frac{1}{n}\sum_{i=1}^{n} \left(m(w_{i},v_{j},\theta) g_{j}(Z_{i})-\hat{\mu}_{j}(\theta)\right)^{2}.
\end{align*}

Our confidence sets are given by
\begin{equation*}
C_{n}(1-\pi)~=~\left\{\theta \in \Theta : T_{n}(\theta)\leq c_{n}(\theta,\pi)\right\},
\end{equation*}
where $c_{n}(\theta,\pi)$ are critical values corresponding to a prespecified significance level $\pi\in(0,1)$. \cite{chernozhukov/chetverikov/kato:2019} propose several options for $c_{n}(\theta,\pi)$, including self-normalized critical values, bootstrap, and two-step hybrid methods. In terms of implementation, the simplest option is the self-normalized critical value, given by
\begin{equation}
c_{n}(\theta,\pi)~=~
\frac{\Phi^{-1}(1-\pi/p_{n})}{\sqrt{1-\left(\Phi^{-1}(1-\pi/p_{n})\right)^{2}/n}},\label{eq:SN_CV}
\end{equation}
where $p_{n}$ is the number of moment inequalities in \eqref{eq:IS_for_CCK}.\footnote{The denominator in \eqref{eq:SN_CV} is positive under the regularity conditions in \cite{chernozhukov/chetverikov/kato:2019}.} This critical value is especially convenient because it does not depend on $\theta$ and depends only on $p_n$, the sample size $n$, and the significance level $\pi$. See \cite{chernozhukov/chetverikov/kato:2019} for the other critical values.

Under suitable regularity conditions on the data-generating process and on the collection of moment inequalities used in the implementation, the inference methods proposed by \cite{chernozhukov/chetverikov/kato:2019} can be used to construct asymptotically valid confidence sets for the parameter of interest. In particular, under the assumptions required by the chosen procedure, for any significance level $\pi \in (0,1)$,
\begin{equation}
\underset{n\to\infty}{\liminf}\ \inf_{P\in\mathcal{P}}\ \inf_{\theta\in\Theta_{I}(P)}
P\left(\theta\in C_{n}(1-\pi)\right)~\geq~1-\pi,
\label{eq:CS_validity}
\end{equation}
where $\mathcal{P}$ denotes a suitable set of data distributions. We do not state this validity result as a separate formal theorem because the precise assumptions required for \eqref{eq:CS_validity} vary with the particular implementation of the methods in \cite{chernozhukov/chetverikov/kato:2019}. Rather, our contribution is to show that the identified set can be characterized by moment inequalities in a form amenable to these procedures.

To conduct inference on parameters derived from the model, such as subvectors of $\theta$ or elasticities, we use projections of the confidence sets.\footnote{One could employ more sophisticated inference methods for these parameters, such as those proposed by \cite{belloni/bugni/chernozhukov:2019}. We nevertheless focus on projection-based inference to keep the exposition simple.} We record this result in Corollary \ref{cor:derived}.

\begin{corollary}\label{cor:derived}
Assume that $C_{n}(1-\pi)$ is an asymptotically valid confidence set for the parameter of interest in the sense of \eqref{eq:CS_validity}. Then, asymptotically valid confidence sets for the various quantities of interest can be obtained by projecting the parameter values in $C_{n}(1-\pi)$. 

In particular, the asymptotically valid confidence set for the own-elasticity of product $j=1,\dots, J$ in a market with observables $w=(\tilde{s},x,p,z)$ is given by:
\begin{equation*}
C_{\mathcal{E}_{jj}}(w; 1-\pi)~=~\left\{
\begin{array}{c}
e_{jj} = -p_{j}\tfrac{\int_{(\zeta ,\nu )}(\alpha +\nu )
\left(\begin{array}{c}
\tfrac{\exp (\beta x_{j} -\alpha p_{j}+\xi _{j}+\zeta x_{j} -\nu p_{j})}{1+\sum_{b=1}^{J}\exp (\beta x_{b} -\alpha p_{b}+\xi _{b}+\zeta x_{b} -\nu p_{b})}\times \\
\left( 1-\tfrac{\exp (\beta x_{j} -\alpha p_{j}+\xi _{j}+\zeta x_{j} -\nu p_{j})}{1+\sum_{b=1}^{J}\exp (\beta x_{b} -\alpha p_{b}+\xi _{b}+\zeta x_{b} -\nu p_{b})} \right)
\end{array}\right)
f(\zeta ,\nu ;\lambda )d(\zeta ,\nu )}{\int_{(\zeta ,\nu )}\tfrac{ \exp \left( \beta x_{j}-\alpha p_{j}+\xi_{j}+ \zeta x_{j} -\nu p_{j} \right) }{1+\sum_{b=1}^{J} \exp \left( \beta x_{b}-\alpha p_{b}+\xi _{b}+\zeta x_{b}  -\nu p_{b} \right) }f( \zeta ,\nu ;\lambda ) d(\zeta ,\nu )}:\\
 \xi  \in \mathcal{U}(w;\theta )~~\text{and}~~\theta \in C_{n}( 1-\pi )
\end{array}
\right\}.
\end{equation*}
In turn, the asymptotically valid confidence set for the cross-price elasticity of product $j=1,\dots, J$ with respect to the price of product $k\neq j$ in a market with observables $w$ is given by:
\begin{equation*}
C_{\mathcal{E}_{jk}}(w; 1-\pi)=~\left\{
\begin{array}{c}
e_{jk} = p_{k}\tfrac{\int_{(\zeta ,\nu )}(\alpha +\nu )
\left(\begin{array}{c}
\tfrac{\exp (\beta x_{j} -\alpha p_{j}+\xi _{j}+\zeta x_{j} -\nu p_{j})}{1+\sum_{b=1}^{J}\exp (\beta x_{b} -\alpha p_{b}+\xi _{b}+\zeta x_{b} -\nu p_{b})}\times \\
\tfrac{\exp (\beta x_{k} -\alpha p_{k}+\xi _{k}+\zeta x_{k} -\nu p_{k})}{1+\sum_{b=1}^{J}\exp (\beta x_{b} -\alpha p_{b}+\xi _{b}+\zeta x_{b} -\nu p_{b})}
\end{array}\right)
f(\zeta ,\nu ;\lambda )d(\zeta ,\nu )}{\int_{(\zeta ,\nu )}\tfrac{ \exp \left( \beta x_{j}-\alpha p_{j}+\xi_{j}+ \zeta x_{j} - \nu p_{j} \right) }{1+\sum_{b=1}^{J} \exp \left( \beta x_{b}-\alpha p_{b}+\xi _{b}+\zeta x_{b}  -\nu p_{b} \right) }f( \zeta ,\nu ;\lambda ) d(\zeta ,\nu )}:\\
 \xi  \in \mathcal{U}(w;\theta )~~\text{and}~~\theta \in C_{n}( 1-\pi )
\end{array}
\right\}.
\end{equation*}
The asymptotically valid confidence set for the markup of product $j=1,\dots, J$ in a market with observables $w$ is given by:
\begin{equation*}
C_{\mathcal{M}_{j}}(w; 1-\pi)~=~\left\{
\begin{array}{c}
M_{j} =  \tfrac{
\int_{(\zeta ,\nu )}\tfrac{ \exp \left( \beta x_{j}-\alpha p_{j}+\xi_{j}+ \zeta x_{j} -\nu p_{j} \right) }{1+\sum_{b=1}^{J} \exp \left( \beta x_{b}-\alpha p_{b}+\xi _{b}+\zeta x_{b}  -\nu p_{b} \right) }f( \zeta ,\nu ;\lambda ) d(\zeta ,\nu )
}{
 \int_{(\zeta ,\nu )}(\alpha +\nu )
\left(\begin{array}{c}
\tfrac{\exp (\beta x_{j} -\alpha p_{j}+\xi _{j}+\zeta x_{j} -\nu p_{j})}{1+\sum_{b=1}^{J}\exp (\beta x_{b} -\alpha p_{b}+\xi _{b}+\zeta x_{b} -\nu p_{b})}\times \\
\left( 1-\tfrac{\exp (\beta x_{j} -\alpha p_{j}+\xi _{j}+\zeta x_{j} -\nu p_{j})}{1+\sum_{b=1}^{J}\exp (\beta x_{b} -\alpha p_{b}+\xi _{b}+\zeta x_{b} -\nu p_{b})} \right)
\end{array}\right)
f(\zeta ,\nu ;\lambda )d(\zeta ,\nu )
}:\\
 \xi  \in \mathcal{U}(w;\theta )~~\text{and}~~\theta \in C_{n}( 1-\pi )
\end{array}
\right\}.
\end{equation*}
Finally, the asymptotically valid confidence set for the diversion ratio of product $j = 1,\dots, J$ with respect to product $k \ne j$ in a market with observables $w$ is given by:
\begin{equation*}
C_{\mathcal{D}_{jk}}(w; 1-\pi)~=~\left\{
\begin{array}{c}
D_{jk} = \tfrac{\int_{(\zeta ,\nu )}(\alpha +\nu )
\left(\begin{array}{c}
\tfrac{\exp (\beta x_{j} -\alpha p_{j}+\xi _{j}+\zeta x_{j} -\nu p_{j})}{1+\sum_{b=1}^{J}\exp (\beta x_{b} -\alpha p_{b}+\xi _{b}+\zeta x_{b} -\nu p_{b})}\times \\
 \tfrac{\exp (\beta x_{k} -\alpha p_{k}+\xi _{k}+x_{k}\zeta -\nu p_{k})}{1+\sum_{b=1}^{J}\exp (\beta x_{b} -\alpha p_{b}+\xi _{b}+\zeta x_{b} -\nu p_{b})} 
\end{array}\right)
f(\zeta ,\nu ;\lambda )d(\zeta ,\nu )}{\int_{(\zeta ,\nu )}(\alpha + \nu) 
\left(\begin{array}{c}
\tfrac{\exp (\beta x_{k} -\alpha p_{k}+\xi _{k}+x_{k}\zeta -\nu p_{k})}{1+\sum_{b=1}^{J}\exp (\beta x_{b} -\alpha p_{b}+\xi _{b}+\zeta x_{b} -\nu p_{b})}\times \\
\left( 1-\tfrac{\exp (\beta x_{k} -\alpha p_{k}+\xi _{k}+x_{k}\zeta -\nu p_{k})}{1+\sum_{b=1}^{J}\exp (\beta x_{b} -\alpha p_{b}+\xi _{b}+\zeta x_{b} -\nu p_{b})} \right)
\end{array}\right)
f(\zeta,\nu;\lambda)d(\zeta,\nu)}:\\
 \xi  \in \mathcal{U}(w;\theta )\text{ and }\theta \in C_{n}( 1-\pi )
\end{array}
\right\}.
\end{equation*}
\end{corollary}

\section{Empirical illustration}\label{sec:Application}

This section presents an empirical illustration based on the data used by GLS, that is, \citet{gandhi/lu/shi:2023}. The data consist of scanner-level observations on tuna products sold by the Dominick's Finer Foods retail chain. Unlike the GLS analysis, we exclude products with zero sales and retain only the top $30$ products.\footnote{This restriction follows \citet{Nevo:2006aa}, who restrict their sample to the top $30$ UPCs and therefore drop products with small or zero sales.} We further restrict attention to markets in which all $30$ products are available. The resulting sample contains $n = 2{,}358$ markets.

In the empirical application in GLS, inside-good market shares are constructed as the ratio of product sales to the number of consumers who visited the store during a given week. Accordingly, the empirical outside share, $s_0$, is defined as one minus the ratio of total quantity sold to the consumer count in the market. The distribution of $s_0$ is concentrated at high values: the minimum is $0.622$, the $25$th percentile is $0.9318$, the median is $0.9457$, the $75$th percentile is $0.9548$, and the maximum is $0.9798$. We use these values to motivate our specification of $S_0(w)$.

The demand model is derived from the following indirect utility specification:
\begin{align*}
u_{ij} ~=~ \beta x_{j} - \alpha p_{j} + \xi_{j} - \nu_{i} p_{j}+ \varepsilon_{ij}, ~~~\text{with}~\nu_{i} ~\sim~ N(0, \lambda^{2}).
\end{align*}
Following the specification in GLS, the observable covariates $x_{j}$ include a constant, a time trend, and a full set of product fixed effects. For each market, we approximate the relevant integrals via simulation by drawing $1{,}000$ individuals from the distribution implied by the model.

The model contains a large number of parameters. Although our inference procedure applies directly to the full parameter vector, doing so in this application would make the computation substantially more demanding and the resulting confidence sets difficult to display. For this reason, we focus on a lower-dimensional implementation. Specifically, we partition the parameters into two groups. The first group contains the parameters of interest, given by the three-dimensional vector $\theta=(\alpha,\beta_1,\lambda)$, where $\alpha$ is the price coefficient, $\beta_1$ is the constant term, and $\lambda$ governs heterogeneity in price sensitivity. The second group contains the remaining parameters, which we treat as nuisance parameters.

We first estimate all model parameters using standard BLP with the observed outside shares, and then fix the nuisance parameters at their estimated values. Given these values, we apply our inference procedure to the lower-dimensional parameter vector $\theta=(\alpha,\beta_1,\lambda)$ without assuming full knowledge of the outside-good shares. In particular, we consider two information structures for the outside-good shares. The first allows for partial information by setting $S_0(w)=[c,1)$. Guided by the empirical distribution of outside shares in the data, we take $c=0.6$ for this exercise. The second allows for a small amount of missing information relative to the benchmark outside share values by setting $S_0(w)=[s_0-0.05,s_0+0.05]\cap(0,1)$, where $s_0$ is the outside-good share calculated by GLS. Under each information structure, we compute $90\%$ confidence sets for $\theta$.


Our inference objective is the three-dimensional parameter vector $\theta = (\alpha,\beta_{1},\lambda)$. For computational purposes, we restrict attention to
\begin{align*}
    \Theta ~=~\Theta_\alpha \times \Theta_{\beta_{1}} \times \Theta_{\lambda}  ~=~ [0, 3] \times [-8, 0] \times [0, 1].
\end{align*}
The parameter space is chosen to contain the standard BLP estimates obtained under the assumption that the observed outside shares are correct, namely $(\hat\alpha,\hat\beta_1,\hat\lambda)=(1.027,-4.1941,0.0363)$.

We construct confidence sets by using the methodology described in Section \ref{sec:Inference}. In particular, we use the class of instrument functions $\mathcal{G}_{cube}$, constructed from the excluded instruments provided in the GLS data. The dataset contains $30$ continuous excluded instruments for each market, and we set $r_0=1$. Constructing $\mathcal{G}_{cube}$ directly, as in Example \ref{ex:hypercube}, would require forming hypercubes in dimension $d_z=30$. With $r_0=1$, each coordinate is partitioned into two intervals, yielding $2^{30}$ cubes. This number is enormous relative to the sample size, $n=2{,}358$. Following the practical recommendation in \citet{andrews/shi:2013}, we therefore use low-dimensional projections of the instrument vector rather than the full $30$-dimensional partition. Specifically, we use selected one- and two-dimensional combinations of the excluded instruments, which yield a final instrument vector of dimension $51$.\footnote{We impose a global cap of 50 instrument columns, excluding the constant, with 60\% allocated to one-dimensional projections and 40\% to two-dimensional projections. Within each projection dimension, subsets of the excluded instruments are considered until the corresponding cap is reached, and columns supported by fewer than three markets are discarded. The constant vector is included, resulting in a total of $51$ instruments.} We discretize $\mathcal{V}$ into $300$ directions. Combined with $51$ instruments, this yields $p_{n} = 15{,}300$ moment inequalities.

We implement the inference methods described in Section \ref{sec:Inference}, including the self-normalized, bootstrap, and two-step hybrid critical values. The resulting confidence sets are qualitatively similar across methods, with the self-normalized procedure yielding the widest sets. For brevity, we therefore report only the self-normalized confidence sets.

\subsection{Partial information on outside shares}

We first report confidence sets for the partial-information case, $S_0(w)=[0.6,1)$. Figure \ref{fig:GLS_CS_partialS0} depicts the $90\%$ confidence set $C_n(0.9)$. Since $C_n(0.9)$ is a three-dimensional object, we visualize it by plotting its sections in the $(\alpha,\beta_1)$ plane for selected values of $\lambda\in[0,1]$. 

\begin{sidewaysfigure}
    \centering
    \includegraphics[width=\textwidth]{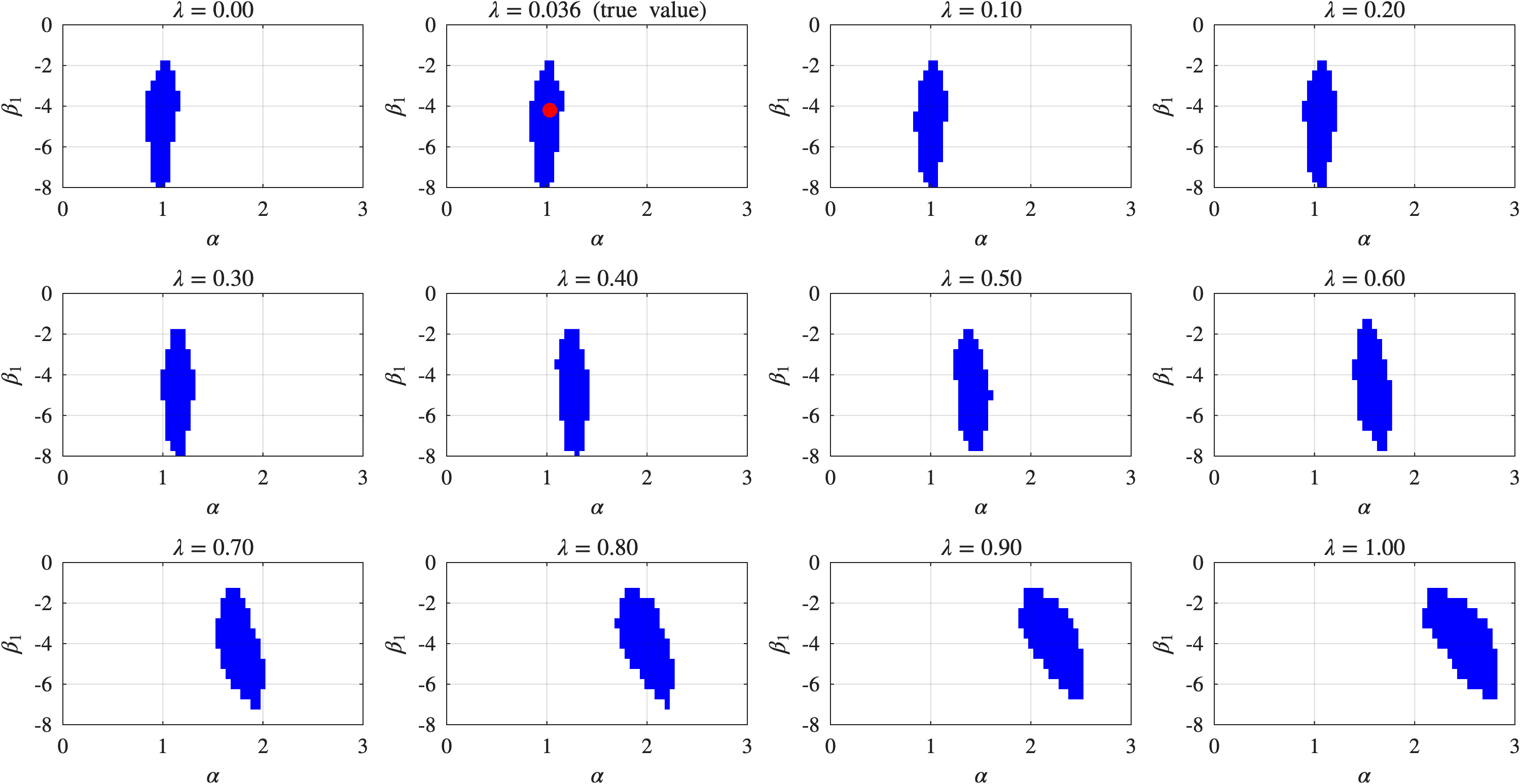}
    \caption{$90\%$ confidence set for $(\alpha,\beta_{1},\lambda)$ and partial information on $s_0$ with $S_0(w)=[0.6,1)$, using the self-normalized critical value.}
    \label{fig:GLS_CS_partialS0}
\end{sidewaysfigure}

The confidence set is admittedly large, especially along the $\beta_1$ and $\lambda$ dimensions. This is a direct reflection of the identifying content of the outside-good shares in the standard BLP analysis. Once full knowledge of these shares is relaxed, the conclusions for some parameters become considerably less precise. Even so, the confidence set remains informative about $\alpha$ and $\beta_1$: it excludes values of $\alpha$ close to zero and rules out values of $\beta_1$ near the upper end of its parameter space. By contrast, the data do not reject any value of $\lambda$ in the parameter space. If we project this confidence set on each parameter coordinate, we obtain the parameter-specific $90\%$ confidence sets:
\begin{align*}
C_\alpha(0.9)~=~[0.85, 2.8],~~~
C_{\beta_{1}}(0.9)~=~[-8, -1.5],~~~\text{and}~~~
C_\lambda(0.9)~=~[0,1].
\end{align*}
As expected, each one of these confidence sets includes their corresponding standard BLP estimates, $(\hat\alpha,\hat\beta_1,\hat\lambda)=(1.027,-4.1941,0.0363)$.

We can also use these results to construct $90\%$ confidence sets for several quantities of interest. For concreteness, we evaluate these objects at the observed values of $w=(\tilde{s},x,p,z)$ in the first market. We begin with confidence sets for elasticities. Since there are $30$ products, many own- and cross-price elasticities could be reported. For conciseness, we focus on three products, selected according to their average market shares. Specifically, we rank products by their average market shares and report results for the first-, second-, and fifteenth-ranked products. This choice allows us to examine interactions between two relatively high-share products, as well as between a high-share product and a medium-share product. The resulting 90\% confidence sets for the selected own- and cross-price elasticities are:
\begin{align*}
&C_{\mathcal{E}_{1,1}}(w;0.9)
= [-1.855,-0.646],
&&C_{\mathcal{E}_{1,2}}(w;0.9)
= [0.000,0.041],
&&C_{\mathcal{E}_{1,15}}(w;0.9)
= [0.000,0.014],\\
&C_{\mathcal{E}_{2,1}}(w;0.9)
= [0.000,0.067],
&&C_{\mathcal{E}_{2,2}}(w;0.9)
= [-1.567,-0.520],
&&C_{\mathcal{E}_{2,15}}(w;0.9)
= [0.000,0.014],\\
&C_{\mathcal{E}_{15,1}}(w;0.9)
= [0.000,0.066],
&&C_{\mathcal{E}_{15,2}}(w;0.9)
= [0.000,0.041],
&&C_{\mathcal{E}_{15,15}}(w;0.9)
= [-2.152,-0.781].
\end{align*}
Although the confidence sets are wide, they preserve several economically meaningful patterns. The own-price elasticities are uniformly negative, the cross-price elasticities are uniformly positive, and the medium-share product (product $15$) tends to exhibit a larger own-price elasticity in magnitude than the two high-share products. As expected, these confidence sets also contain the corresponding elasticity values evaluated at the standard BLP estimates: $e_{1,1} = -0.808$, $e_{2,2} = -0.645$, $e_{15,15} = -0.954$, $e_{1,2} = 0.002$, $e_{1,15} = 0.001$, $e_{2,1} = 0.003$, $e_{2,15} = 0.001$, $e_{15,1}=0.003$ and $e_{15,2} = 0.002$.

The 90\% confidence sets for markups and diversion ratios are:
\begin{align*}
&C_{\mathcal{M}_{1}}(w;0.9)
= [0.426,1.222],
&&C_{\mathcal{M}_{2}}(w;0.9)
= [0.402,1.211],
&&C_{\mathcal{M}_{15}}(w;0.9)
= [0.432,1.191],\\
&C_{\mathcal{D}_{1,2}}(w;0.9)
= [0.000,0.038],
&&C_{\mathcal{D}_{1,15}}(w;0.9)
= [0.000,0.037],
&&C_{\mathcal{D}_{2,1}}(w;0.9)
= [0.000,0.030],\\
&C_{\mathcal{D}_{2,15}}(w;0.9)
= [0.000,0.029],
&&C_{\mathcal{D}_{15,1}}(w;0.9)
= [0.000,0.007],
&&C_{\mathcal{D}_{15,2}}(w;0.9)
= [0.000,0.007].
\end{align*}
The corresponding values evaluated at the standard BLP estimates are $M_{1}=0.978$, $M_{2}=0.977$, and $M_{15}=0.975$ for markups, and $D_{1,2}=0.004$, $D_{1,15}=0.004$, $D_{2,1}=0.003$, $D_{2,15}=0.003$, $D_{15,1}=0.001$, and $D_{15,2}=0.001$ for diversion ratios.
The markup confidence sets are fairly wide, indicating substantial uncertainty about firms' pricing power once full knowledge of the outside-good shares is relaxed. The diversion ratio confidence sets are close to zero, reflecting the small market shares of these products in the selected market.

Overall, this exercise illustrates that uncertainty about the share of the outside goods translates into substantial uncertainty about demand levels and preference heterogeneity. In particular, the confidence set is wide along the $\beta_1$ and $\lambda$ dimensions, and the markup confidence sets are correspondingly broad. At the same time, the confidence sets for diversion ratios and market shares remain close to zero for the selected products. Thus, in this application, limited information on market size primarily affects conclusions about demand levels and markups, whereas the evidence of limited substitution among the selected products is comparatively robust.

\subsection{Small amount of missing information on the outside shares}

We next consider the case in which there is only a small amount of missing information about the outside-good share, with $S_0(w)=[s_0-0.05,s_0+0.05]\cap(0,1)$. Figure \ref{fig:GLS_CS_smallS0} depicts the resulting $90\%$ confidence set $C_n(0.9)$. Since the lowest empirical outside share is $0.622$ and we set $\varepsilon=0.05$, this specification provides substantially more information about the outside-good shares than the one considered in the previous subsection. As a consequence, the confidence set in Figure \ref{fig:GLS_CS_smallS0} is considerably tighter than the one in Figure \ref{fig:GLS_CS_partialS0}.

\begin{sidewaysfigure}
    \centering
    \includegraphics[width=\textwidth]{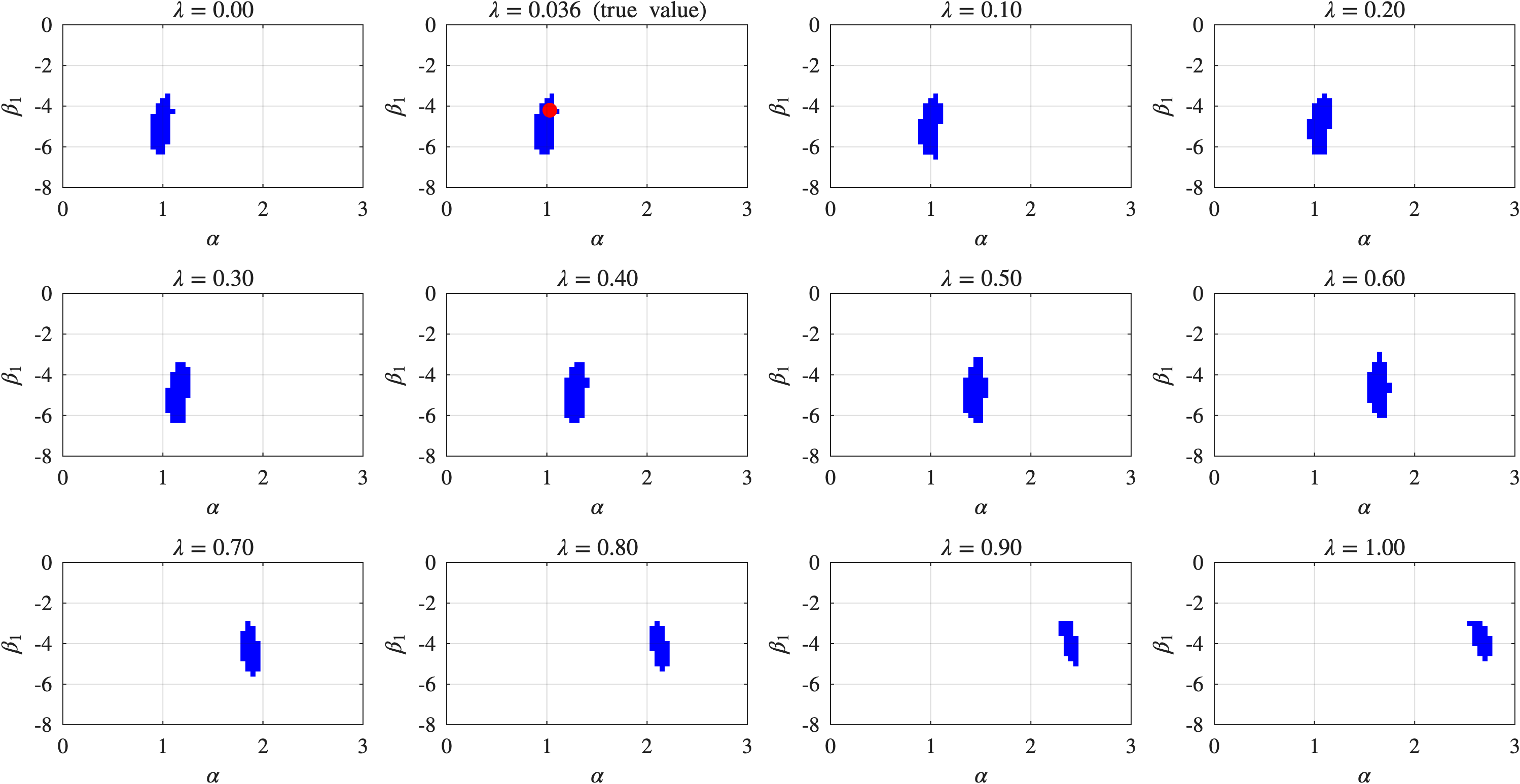}
    \caption{90\% confidence set for $(\alpha,\beta_{1},\lambda)$ and a small amount of missing information on $s_0$, given by $S_0(w)=[s_0-0.05,s_0+0.05]\cap(0,1)$, using the self-normalized critical value.}
    \label{fig:GLS_CS_smallS0}
\end{sidewaysfigure}

Projecting the confidence set onto each parameter coordinate yields parameter-specific confidence sets: 
\begin{align*}
C_\alpha(0.9)~=~[0.9,2.75],~~~
C_{\beta_{1}}(0.9)~=~[-6.5,-3],~~~\text{and}~~~
C_\lambda(0.9)~=~[0,1].
\end{align*}
Despite the missing information on the share of the outside goods, the confidence intervals are quite informative about $\alpha$ and $\beta_1$. By contrast, the projected confidence set for $\lambda$ covers the entire parameter space, indicating that the heterogeneity parameter is not well restricted in this application. This is consistent with the limited precision of the preference-heterogeneity parameters under standard BLP estimation in this empirical setting.

Finally, we compute confidence sets for several quantities of interest. As before, we evaluate these objects at the observed values of $w $ in the first market and focus on the first-, second-, and fifteenth-ranked products by average market share.
\begin{align*}
&C_{\mathcal{E}_{1,1}}(w;0.9)
= [-1.658,-0.700],
&&C_{\mathcal{E}_{1,2}}(w;0.9)
= [0.000,0.009],
&&C_{\mathcal{E}_{1,15}}(w;0.9)
= [0.000,0.003],\\
&C_{\mathcal{E}_{2,1}}(w;0.9)
= [0.000,0.015],
&&C_{\mathcal{E}_{2,2}}(w;0.9)
= [-1.412,-0.560],
&&C_{\mathcal{E}_{2,15}}(w;0.9)
= [0.000,0.003],\\
&C_{\mathcal{E}_{15,1}}(w;0.9)
= [0.000,0.015],
&&C_{\mathcal{E}_{15,2}}(w;0.9)
= [0.000,0.009],
&&C_{\mathcal{E}_{15,15}}(w;0.9)
= [-1.856,-0.829].
\end{align*}
Relative to the partial-information case, these confidence sets are narrower, reflecting the additional information about the outside-good shares. The qualitative patterns, however, are unchanged. The own-price elasticity confidence sets remain uniformly negative, the cross-price elasticity confidence sets remain nonnegative and close to zero, and product $15$ continues to exhibit the largest own-price elasticity in magnitude among the selected products. As expected, the confidence sets also contain the corresponding elasticity values evaluated at the standard BLP estimates reported above.

In this case, the $90\%$ confidence sets for markups and diversion ratios are:
\begin{align*}
&C_{\mathcal{M}_{1}}(w;0.9)
= [0.476,1.128],
&&C_{\mathcal{M}_{2}}(w;0.9)
= [0.446,1.124],
&&C_{\mathcal{M}_{15}}(w;0.9)
= [0.501,1.122],\\
&C_{\mathcal{D}_{1,2}}(w;0.9)
= [0.000,0.008],
&&C_{\mathcal{D}_{1,15}}(w;0.9)
= [0.000,0.010],
&&C_{\mathcal{D}_{2,1}}(w;0.9)
= [0.000,0.007],\\
&C_{\mathcal{D}_{2,15}}(w;0.9)
= [0.000,0.007],
&&C_{\mathcal{D}_{15,1}}(w;0.9)
= [0.000,0.002],
&&C_{\mathcal{D}_{15,2}}(w;0.9)
= [0.000,0.001].
\end{align*}
Once again, our confidence sets are somewhat narrower than those obtained in the previous subsection. The markup confidence sets remain fairly wide, but their lower and upper bounds are both closer to the corresponding standard BLP values. The confidence sets for diversion ratios are also much narrower and remain close to zero, reflecting the small market shares of the selected products.

The comparison with the partial-information case highlights the value of even limited information about the share of the outside goods. Allowing $s_0$ to vary only in a neighborhood of the empirical outside share substantially tightens the confidence sets for $\beta_1$, elasticities, and diversion ratios, while leaving the broad economic conclusions unchanged. The main exception is $\lambda$, whose confidence set continues to cover the full parameter space. This suggests that, in this application, local information about market size is useful for sharpening conclusions about demand levels and substitution patterns, but the data remain much less informative about heterogeneity in price sensitivity.

\section{Conclusions}\label{sec:Conclusions}

This paper studies identification and inference in the BLP demand model when the share of the outside good is unobserved or only partially known. We consider a framework in which the researcher specifies a set that contains the share of the outside good. We show that relaxing the standard assumption of a known outside share typically leads to partial identification, and we derive sharp identified sets for both structural parameters and economically relevant quantities. In particular, we characterize sharp identified sets for structural parameters, unobserved demand shocks, and quantities of interest such as own- and cross-price elasticities, equilibrium markups, and diversion ratios. 

The framework can be used in several ways. When the researcher has partial information about market size, this yields a set of parameter values and counterfactual objects consistent with that information. When the researcher has a benchmark value for the outside share, it can be used for sensitivity analysis by considering neighborhoods around that value. When the researcher has little or no prior information, the approach remains applicable by setting the admissible set for the outside share equal to the open unit interval.

We also develop inference procedures that construct uniformly valid confidence sets for these objects in this partially identified setting. Our results show that even when point identification fails, the model can still yield informative conclusions about demand and market outcomes. Overall, the framework provides a basis for demand estimation and counterfactual analysis when market size is uncertain or only partially known.

Our empirical application illustrates how the proposed framework can be used in practice. It shows how uncertainty about market size can be incorporated directly into demand estimation and used to complement benchmark specifications based on a single outside-share value. The resulting identified sets quantify the extent to which empirical conclusions are driven by information in the model, as opposed to the particular assumptions imposed on the outside-good share.

At the same time, our results show that the identified sets, and the corresponding confidence sets, can be large, especially along certain dimensions. Because our identified sets are sharp, large identified sets indicate that conclusions from the BLP model can be sensitive to the information available about the outside share. Similarly, large confidence sets reflect the combined effect of sampling uncertainty and the loss of identifying information that arises when the outside-share assumption is weakened. In this sense, the analysis quantifies the identifying content of the standard outside-share assumption. When that assumption is relaxed, the data and the model may contain substantially less information about some structural parameters or counterfactual objects. This highlights the value of making assumptions about market size and the outside-good share explicit in empirical IO applications, and of assessing how sensitive empirical conclusions are to those assumptions.

\appendix
 
\begin{small}
 \section{Appendix}\label{sec:appendix_proofs}

 \subsection{Proofs of theorems}
\begin{proof}[Proof of Theorem \ref{thm:Id_set_expression}]
This result is based on \citet[Theorem 5]{chesher/rosen:2017} with $Y=( \tilde{s},x,p) $, $Z=z \in \mathbb{S}_z$, $U=\xi$, and $\mathcal{C}=\{ 0\} $. Recall that $w=(\tilde{s},x,p,z)$. We first verify their Assumptions A1-A5 and MI. To this end, we restate each one of their conditions in our notation and briefly explain why it holds.
\begin{enumerate}
\item Assumption A1: $( w,\xi ) $ are random variables defined on the probability space $( \Omega ,\mathcal{B},P)$, where $\mathcal{B}$ is the Borel $\sigma$-algebra. The support of $(w,\xi ) $ is a subset of a finite-dimensional Euclidean space.\\This condition holds in our setup with $(w,\xi ) \in \Delta _{J}\times \mathbb{R}^{d_{x}}\times \mathbb{R}^{d_{p}}\times \mathbb{R}^{d_{z}}\times \mathbb{R}^{J}$.

\item Assumption A2: The collection of conditional distributions $\mathcal{F} _{( \tilde{s},x,p) |z}=\{ F_{( \tilde{s},x,p) |z=\varpi }:\varpi \in \mathbb{S}_{z}\} $ is identified by the sampling process, where for all $A\subseteq S_{( \tilde{s},x,p) |z=\varpi } $, $\mathcal{F}_{( \tilde{s},x,p) |z}( A|\varpi ) =P( (\tilde{s},x,p)\in A|z=\varpi ) $.\\This follows from the fact that the data are assumed an i.i.d.\ sample from $( \tilde{s} ,x,p,z) $.

\item Assumption A3: There is a $\mathcal{B}$-measurable function $h$ such that $P( h( w,\xi;\theta )=0) =1$, and there is a collection of conditional distributions $\mathcal{G}_{\xi |z}=\{ G_{\xi |z}( \cdot |\varpi ) :\varpi \in \mathbb{S}_{z}\} $, where for all $A\subseteq \mathcal{S}_{\xi |\varpi }$, $G_{\xi |z}( A|\varpi ) =P( \xi \in A|z=\varpi ) $.\\
This holds by defining:
\begin{equation}
h(w,\xi ;\theta) ~=~
\left\{\begin{array}{c}
\left\Vert \tilde{s}-
\dfrac{
\{ \int_{(\zeta ,\nu )}\frac{\exp ( \beta x_{j}-\alpha p_{j}+\xi _{j}+\zeta x_{j} -\nu p_{j} ) }{1+\sum_{b=1}^{J}\exp ( \beta x_{b}-\alpha p_{b}+\xi _{b}+\zeta x_{b} -\nu p_{b} ) }f( \zeta ,\nu ;\lambda ) d( \zeta ,\nu ) \} _{j=1}^{J}
}
{
\int_{(\zeta ,\nu )}\frac{\sum_{e=1}^{J}\exp ( \beta x_{e}-\alpha p_{e}+\xi _{e}+\zeta x_{e}-\nu p_{e} ) }{1+\sum_{l=1}^{J}\exp ( \beta x_{l}-\alpha p_{l}+\xi _{l}+\zeta x_{l} -\nu p_{l} ) }f( \zeta ,\nu ;\lambda ) d( \zeta ,\nu ) 
}\right\Vert\\
+~ I[\sigma_0\big((\beta x_j-\alpha p_j+\xi_j)_{j=1}^J,x,p;\lambda\big) \not\in S_0(w)]
\end{array}\right\},
\label{eq:h_defn}
\end{equation}
where $\sigma_0(\delta,x,p;\lambda)$ is as in \eqref{eq:sigma0_defn}. With this definition, $h(w,\xi;\theta)=0$ if and only if $\xi\in\mathcal{U}(w;\theta)$. The measurability of $h$ follows from the measurability of the share functions and the set $S_0(w)$.

\item Assumption A4: The pair $(h,\mathcal{G}_{\xi |z})$ belongs to a known set of admissible structures $\mathcal{M}$.\\
In our case, we have $\mathcal{M}=\mathcal{M}_{1}\times \mathcal{M}_{2}$, where $\mathcal{M}_{1}=\{ h(\cdot;\theta):\theta \in \Theta \} $ with $h$ as in \eqref{eq:h_defn} and $\mathcal{M}_{2}$ is the space of conditional distributions that satisfy $E[ \xi |z] =0$ a.e.\ $z\in \mathbb{S}_{z}$, respectively. The latter condition is imposed by the fact that $z$ is an IV in this problem.

\item Assumption A5: The random set $\mathcal{U}(w; \theta) $ in \eqref{eq:U_defn} is closed a.s.\ $P(\cdot |z)$ for each $z\in \mathbb{S}_{z}$.\\This is shown in Lemma \ref{lem:U_is_Closed}.

\item Assumption MI with $\mathcal{C}=\{ 0\} $: $G_{\xi |z}$ satisfies $E[ \xi |z ] =0$ a.e.\ $z \in \mathbb{S}_{z}$.\\
The condition follows from the fact that $z$ is an IV in this problem.
\end{enumerate}
Under these conditions, \citet[Theorem 5]{chesher/rosen:2017} implies that
\begin{equation}
\Theta _{I}(P) ~=~\big\{ \theta \in \Theta:\mathbf{0}\in \mathbb{E}[ \mathcal{U}(w; \theta) |z ] \text{ a.e. }z \in \mathbb{S}_{z}\big\} ,
\label{eq:theta_I_Expr1}
\end{equation}
where $\mathbb{E}[\cdot|z]$ denotes the conditional selection or Aumann expectation, see \citet[Section 1.6]{molchanov:2005} or \citet[Definition 3.6]{molchanov/molinari:2018}.
Under our integrability assumption, the conditional selection expectation in \eqref{eq:theta_I_Expr1} is well-defined. By the conditional non-atomicity assumption, this conditional selection expectation is convex a.s.; see \citet[Theorem 3.7]{molchanov/molinari:2018}. Thus, by the support-function characterization of closed convex sets and the expected-support-function formula for selection expectations, $\mathbf{0}\in \mathbb{E}[ \mathcal{U}(w; \theta) |z] $ if and only if
\begin{equation}
    E[ \sup\nolimits_{\xi \in \mathcal{U}(w; \theta) }v^{\prime }\xi |z] \geq 0~~~~\text{for all $v\in \mathbb{R}^{J}$ with $\Vert v\Vert =1$ and a.e.\ $z \in \mathbb{S}_{z}$}
    \label{eq:support_ineq_U}
\end{equation}
where we allow the conditional expectation in \eqref{eq:support_ineq_U} to be extended-valued; see \citet[Theorem 3.11]{molchanov/molinari:2018}. In particular, if the left-hand side of \eqref{eq:support_ineq_U} is equal to $+\infty$ in some direction $v$, then the corresponding inequality is automatically satisfied and imposes no restriction on $\theta$. By this and \eqref{eq:theta_I_Expr1}, we obtain
\begin{equation}
\Theta _{I}(P) ~=~\big\{ 
\theta \in \Theta:E[ \sup\nolimits_{\xi \in \mathcal{U}(w; \theta) }v^{\prime }\xi |z] \geq 0 ~~~
\text{for all }v\in \mathbb{R}^{J}:\Vert v\Vert =1\text{ and a.e. }z \in \mathbb{S}_{z}
\big\} .
\label{eq:theta_I_Expr2}
\end{equation}

To establish \eqref{eq:Id_set_expression} from \eqref{eq:theta_I_Expr2}, it suffices to show that for $\theta =( \alpha ,\beta ,\lambda ) $,
\begin{equation}
\sup\nolimits_{\xi \in \mathcal{U}(w; \theta) }v^{\prime }\xi ~=~\sup\nolimits_{s_{0}\in S_0(w)}v^{\prime }( \sigma ^{-1}( \tilde{s}(1-s_{0}),x,p;\lambda ) -( \beta x_{j}-\alpha p_{j}) _{j=1}^{J}) .
\label{eq:theta_I_aux}
\end{equation}
In turn, this follows from $\mathcal{U}(w; \theta) = \cup _{s_{0}\in  S_0(w)}\{ \sigma ^{-1}(\tilde{s}(1-s_{0}),x,p;\lambda) -( \beta x_{j}-\alpha p_{j}) _{j=1}^{J} \}$, as shown in Lemma \ref{lem:charaU}.
\end{proof}

\begin{proof}[Proof of Theorem \ref{thm:id_sets1}]
Let $\tilde{\Theta} _{I}(P)$ denote the set in the right-hand side expression. We divide the proof into parts. Part 1 shows that $\Theta _{I}(P)\subseteq \tilde{\Theta} _{I}(P)$, and part 2 shows that $\tilde{\Theta} _{I}(P)\subseteq \Theta _{I}(P)$. Their combination implies the desired result.

\noindent \underline{Part 1: $\Theta _{I}(P)\subseteq \tilde{\Theta} _{I}(P)$.} Fix $\theta \in \Theta _{I}(P)$ arbitrarily, and we want to show $\theta \in \tilde{\Theta} _{I}(P)$. By $\theta \in \Theta _{I}(P)$ and the arguments in the proof of Theorem \ref{thm:Id_set_expression}, $\mathbf{0}_{J\times 1}\in \mathbb{E}[\mathcal{U}(w;\theta )|z]$ a.e.\ $z\in \mathbb{S}_{z}$, where $\mathbb{E} [\cdot |z]$ denotes the conditional selection or Aumann expectation. By definition and the fact that $\left\{ \left( \tilde{s},x,p\right) |z\right\} $ has a non-atomic probability space, this implies the existence of a measurable selection $\zeta $ of $\mathcal{U}(w;\theta )$ such that 
\begin{equation}
\mathbf{0}_{J\times 1}~=~E \left[ \zeta |z\right] ~\text{a.e.}~z\in \mathbb{S} _{z}.
\label{eq:id_sets_1}
\end{equation}
By this and Lemma \ref{lem:charaU}, we conclude that $\exists s_{0}^*:\mathbb{S}_{w}\to (0,1)$ with $s_0^*(w) \in S_0(w)$ for all $w \in \mathbb{S}_w$ such that 
\begin{equation}
\zeta ~=~\sigma ^{-1}\left( \tilde{s}(1-s^*_{0}(w)),x,p;\lambda \right) -\left( \beta x_{j}-\alpha p_{j}\right) _{j=1}^{J}.
\label{eq:id_sets_2}
\end{equation}
By \eqref{eq:id_sets_1} and \eqref{eq:id_sets_2}, it follows that $\theta \in \tilde{\Theta} _{I}(P)$, as desired.

\noindent \underline{Part 2: $\tilde{\Theta} _{I}(P)\subseteq \Theta _{I}(P)$.} Fix $\theta \in \tilde{\Theta} _{I}(P)$ arbitrarily, and we want to show $\theta \in \Theta _{I}(P)$. By $\theta \in \tilde{\Theta} _{I}(P)$, $\exists s^*_{0}:\mathbb{S}_{w}\to (0,1)$ with $s^*_0(w) \in S_0(w)$ for all $w \in \mathbb{S}_w$ such that 
\begin{equation}
E[\sigma ^{-1}\left( \tilde{s}(1-s^*_{0}(w)),x,p;\lambda \right) -(\beta x_{j}-\alpha p_{j})_{j=1}^{J}|z]~=~\mathbf{0}_{J\times 1}~\text{a.e.}~z\in \mathbb{S}_{z}.
\label{eq:id_sets_3}
\end{equation}
For each $w\in \mathbb{S}_{w}$ and $v\in \mathbb{R}^{J}$ with $\lVert v\rVert =1$, $s_0:=s^*_0(w)\in S_{0}(w)$ implies
\begin{align}
\sup_{s_0\in S_{0}(w)}v^{\prime }(\sigma ^{-1}\left(\tilde{s} (1-s_0),x,p;\lambda \right) -(\beta x_{j}-\alpha p_{j})_{j=1}^{J})~\geq~ v^{\prime }(\sigma ^{-1}\left( \tilde{s}(1-s_0),x,p;\lambda \right) -(\beta x_{j}-\alpha p_{j})_{j=1}^{J}).
\label{eq:id_sets_4}
\end{align}
By taking conditional expectations, and using \eqref{eq:id_sets_3} and \eqref{eq:id_sets_4}, we have that $ \theta \in \Theta _{I}(P)$, as desired.
\end{proof}

\begin{proof}[Proof of Theorem \ref{thm:id_sets2}]
Per the statement, $S_0(w)=S_0(z)$ for all $w \in \mathbb{S}_w$. We divide the proof into parts. Part 1 shows that $\mathcal{H}(P)\subseteq  \Theta _{I}(P)$, and part 2 provides an example of a data-generating process with ${\Theta} _{I}(P)\neq \emptyset$ and $\mathcal{H}(P)= \emptyset$.

\noindent \underline{Part 1: $\mathcal{H}(P)\subseteq \Theta _{I}(P)$.} This follows immediately from Theorem \ref{thm:id_sets1} and that any function $s_{0}:\mathbb{S}_{z}\to (0,1)$ with $s_0(z) \in S_0(z)$ for all $z \in \mathbb{S}_z$ is a special case of a function $ s_{0}:\mathbb{S}_{w}\to (0,1)$ with $s_0(w) \in S_0(w) = S_0(z)$  for all $w \in \mathbb{S}_w$.

\noindent \underline{Part 2: There are DGPs with ${\Theta} _{I}(P)\neq \emptyset$ and $\mathcal{H}(P)= \emptyset$.} For simplicity, we consider an example with $S_0(z) = (0,1)$. By Lemma \ref{lem:rand_coeff_necessary}, the example then requires a BLP model with random coefficients (i.e., $\lambda \neq \bar{\lambda}$).

Consider a DGP with $J=2$ inside products. We assume $z=0$, $x_{1}=0$, $x_{2}\sim 10B$, $B\sim {\rm Bernoulli}(1/2)$, $p$ has any continuous distribution (e.g., $p=(p_{1},p_{2})$ with $p_{1},p_{2}$ i.i.d.\ $U(0,1)$), independent of $x = (x_1,x_2)$, and $\xi \sim N(\mathbf{0}_{2\times 1},\mathbf{I}_{2\times 2})$, independent of $(x,p)$. The random coefficients are given by $\nu =0$ and $\zeta \sim N(0,\lambda)$, independent of $(x,p,\xi)$. The coefficients are $\theta =(\alpha ,\beta ,\lambda ) =(0,0,1) \in \Theta$. For simplicity, we set $\Theta =\{(0,0,1)\}$. The conditional shares satisfy
\begin{align*}
(s_0, s_{1} ,s_2)~=~
\int \Big(  \tfrac{1}{1+\exp (\xi _{1})+\exp (\xi _{2}+\zeta x_{2})}, \tfrac{\exp (\xi _{1})}{1+\exp (\xi _{1})+\exp (\xi _{2}+\zeta x_{2})} ,
\tfrac{\exp (\xi _{2}+\zeta x_{2})}{1+\exp (\xi _{1})+\exp (\xi _{2}+\zeta x_{2})} \Big)\phi(\zeta) d\zeta
\end{align*}
Also, set $\tilde{s}=( s_{1}/(1-s_{0}) ,s_{2}/(1-s_{0}) )$. For any $\delta \in \mathbb{R}^{2}$, define 
\begin{equation}
\sigma (\delta ,x,p;\lambda )~=~\int\left( 
\begin{array}{c}
 \frac{\exp (\delta _{1})}{1+\exp (\delta _{1})+\exp (\delta _{2}+\zeta x_{2})}  , 
 \frac{\exp (\delta _{2}+\zeta x_{2})}{1+\exp (\delta _{1})+\exp (\delta _{2}+\zeta x_{2})}
\end{array}
\right) \phi(\zeta) d\zeta .
\end{equation}
Note that the right-hand side does not depend on $p$ or $x_{1}$, by $\alpha =0$ and $x_{1}=0$. For any target vector $s\in \Delta _{3}$ with $s_{j}>0$
for $j=1,2,3,$ we can compute $\sigma ^{-1}( (s_1,s_2),x,p;\lambda )$, which will not depend on $p$ or $x_{1}$ (by $\alpha =0$ and $x_{1}=0$). The
conditional inside shares are
\begin{align*}
\tilde{s}_{1} ~=~\frac{\int \frac{\exp (\xi _{1})}{1+\exp (\xi _{1})+\exp (\xi _{2}+\zeta x_{2})}\phi(\zeta) d\zeta }{1-\int \frac{1}{ 1+\exp (\xi _{1})+\exp (\xi _{2}+\zeta x_{2})}\phi(\zeta) d\zeta } ~~\text{ and }~~
\tilde{s}_{2} ~=~\frac{\int \frac{\exp (\xi _{2}+\zeta x_{2})}{1+\exp (\xi _{1})+\exp (\xi _{2}+\zeta x_{2})}\phi(\zeta) d\zeta }{1-\int 
\frac{1}{1+\exp (\xi _{1})+\exp (\xi _{2}+\zeta x_{2})}\phi(\zeta) d\zeta }.
\end{align*}
Note that $\{ (\tilde{s},x,p) |z\} $ is non-atomic, as $p$ is continuously distributed.

By construction, $\theta \in \Theta _{I}(P)$. We now show that $\theta \not\in \mathcal{H}(P)$. To see why, note that $\theta \in \mathcal{H}(P)$ would require $\exists
s_{0}^{\ast }:\mathbb{S}_{z}\rightarrow (0,1)$ s.t.
\begin{equation*}
E[\sigma _{j}^{-1}(\tilde{s}(1-s_{0}^{\ast }(z)),x,p;\lambda ) -(\beta x_{j}-\alpha p_{j})|z]=0\text{ for }j=1,2.
\end{equation*}
Since $\Theta =\{(0,0,1)\}$, $z=0$, $x_{1}=0$, $x_{2}\sim 10B$ with $B\sim {\rm Bernoulli}(1/2)$, and $x_{2}\perp \xi $, this is
equivalent to finding a constant $s_{0}^{\ast }\in (0,1)$ s.t.
\begin{align*}
&\mathbf{0}_{2\times 1}~=~\\
&\left( 
\begin{array}{c}
E_{\xi }\left[ \sigma ^{-1}\left( \left(\frac{\exp (\xi _{1})(1-s_{0}^{\ast })}{\exp (\xi _{1})+\exp (\xi _{2})},\frac{\exp (\xi _{2})(1-s_{0}^{\ast })}{\exp (\xi _{1})+\exp (\xi _{2})}\right),(0,0),p;1\right)  \right] + \\ 
E_{\xi }\left[ \sigma ^{-1}\left( \left(\frac{\left( \int \frac{\exp (\xi _{1})}{1+\exp (\xi _{1})+\exp (\xi _{2}+10\zeta )}\phi(\zeta) d\zeta \right) (1-s_{0}^{\ast })}{\left( 1-\int \frac{1}{1+\exp (\xi _{1})+\exp (\xi _{2}+10\zeta )}\phi(\zeta) d\zeta \right) }, \frac{\left( \int \frac{\exp (\xi _{2}+10\zeta )}{1+\exp (\xi _{1})+\exp (\xi _{2}+10\zeta )}\phi(\zeta) d\zeta \right) (1-s_{0}^{\ast })}{ \left( 1-\int \frac{1}{1+\exp (\xi _{1})+\exp (\xi _{2}+10\zeta )}\phi(\zeta) d\zeta \right) }\right),(0,10),p;1\right) \right] 
\end{array}
\right),
\end{align*}
where $E_{\xi }$ denotes the expectation over $\xi $. We can (numerically) show this is impossible. Define the function:
\begin{align*}
&F(s_{0}) ~=~\\
&\sum_{j=1}^{2}\left( 
\begin{array}{c}
E_{\xi }\left[ \sigma ^{-1}\left( \left(\frac{\exp (\xi _{1})(1-s_{0})}{\exp (\xi _{1})+\exp (\xi _{2})},\frac{\exp (\xi _{2})(1-s_{0})}{\exp (\xi _{1})+\exp (\xi _{2})}\right),(0,0),p;1\right) \right] + \\ 
E_{\xi }\left[ \sigma ^{-1}\left(\left(\frac{\left( \int \frac{\exp (\xi _{1})}{1+\exp (\xi _{1})+\exp (\xi _{2}+10\zeta )}\phi(\zeta) d\zeta \right) (1-s_{0})}{\left( 1-\int \frac{1}{1+\exp (\xi _{1})+\exp (\xi _{2}+10\zeta )}\phi(\zeta) d\zeta \right) }, \frac{\left( \int \frac{\exp (\xi _{2}+10\zeta )}{1+\exp (\xi _{1})+\exp (\xi _{2}+10\zeta )}\phi(\zeta) d\zeta \right) (1-s_{0})}{ \left( 1-\int \frac{1}{1+\exp (\xi _{1})+\exp (\xi _{2}+10\zeta )}\phi(\zeta) d\zeta \right) }\right),(0,10),p;1\right) \right] 
\end{array}
\right) ^{2}.
\end{align*}
We can verify that $F(s_{0})$ attains a minimizer at $s_{0}\approx 0.23$, and its minimized value is $F(s_{0})\approx 0.179>0$.    
\end{proof}

\begin{proof}[Proof of Theorem \ref{thm:eq_objects}]
    This result follows from Theorem \ref{thm:Id_set_expression} and the derivations for the elasticities and markup in Section \ref{sec:standard_BLP}. For brevity, we focus on the argument for the own-price elasticity of product $j=1,\dots,J$, as the arguments for the other quantities are completely analogous.
    
    Consider any $e_{jj}\in\mathcal{E}_{jj}(w)$. We show that $e_{jj}$ belongs to the identified set for the own-price elasticity of product $j$ in the BLP model. By definition of $\mathcal{E}_{jj}(\tilde{s},x,p)$, there exist $\theta\in\Theta_I(P)$ and $\xi\in\mathcal{U}(w;\theta)$ such that $e_{jj}$ satisfies \eqref{eq:elasticity_own2}, i.e., is the own-price elasticity of product $j$ implied by $(\theta,\xi)$. It therefore suffices to verify that $(\theta,\xi)$ is a feasible configuration of structural parameters and demand shocks. By Theorem \ref{thm:Id_set_expression}, $\Theta_I(P)$ is precisely the set of parameter values compatible with the BLP model and the data distribution $P$, and for any $\theta\in\Theta_I(P)$, the set $\mathcal{U}(w;\theta)$ consists of the demand-shock vectors that rationalize the inside shares under $\theta$. Hence, $(\theta,\xi)$ is fully compatible with the data and represents a plausible configuration of structural parameters and demand shocks. Thus, $e_{jj}$ belongs to the identified set for the own-price elasticity of product $j$ in the BLP model.

    Conversely, let $e_{jj}$ be any value in the identified set for the own-price elasticity of product $j$ in the BLP model. According to this model, there exists a structural parameter $\theta$ and a vector of demand shocks $\xi$ that are compatible with the data and rationalize the value $e_{jj}$. By Theorem \ref{thm:Id_set_expression}, this compatibility is equivalent to $\theta\in\Theta_I(P)$ and $\xi\in\mathcal{U}(w;\theta)$. These conditions, together with the equation for the own-price elasticity in \eqref{eq:elasticity_own2}, are exactly the membership conditions for $e_{jj}\in\mathcal{E}_{jj}(w)$. This establishes the converse and completes the proof.
\end{proof}

\subsection{Auxiliary results}

\begin{lemma}\label{lem:BLP_probs} 
Under our assumptions, \eqref{eq:BLP_probs_1} holds.
\end{lemma}
\begin{proof}
To this end, fix $j=1,\dots ,J$ arbitrarily. For any $b=1,\dots ,J$, let $\delta _{b}=\beta x_{b} -\alpha p_{b}+\xi _{b}$ and $u_{b}=\zeta x_{b} -\nu p_{b}$, and set $\delta _{0}=0$. Then, 
\begin{align*}
& P(y=j|x,p,\xi ;\theta ) \\
& = ~\int_{(\zeta ,\nu)} P(y=j|x,p,\xi ,\zeta ,\nu;\theta)  dP(\zeta ,\nu \mid x,p,z,\xi;\theta) \\
& \overset{(1)}{=}~ \int_{(\zeta ,\nu)} P\bigg( \delta _{j}+u_{j}+\epsilon _{j} \geq \max_{b \in \{0,\dots ,J\} \setminus \{j\}} \{ \delta _{b}+u_{b}+\epsilon _{b} \} \,\bigg|\, x,p,z,\xi ,\zeta ,\nu ;\theta \bigg)  dP(\zeta ,\nu \mid x,p,z,\xi) \\
& \overset{(2)}{=}~ \int_{(\zeta ,\nu)} P\bigg( \delta _{j}+u_{j}+\epsilon _{j} \geq \max_{b \in \{0,\dots ,J\} \setminus \{j\}} \{ \delta _{b}+u_{b}+\epsilon _{b} \} \,\bigg|\, x,p,z,\xi ,\zeta ,\nu;\theta \bigg) f(\zeta ,\nu ;\lambda) d(\zeta ,\nu) \\
& \overset{(3)}{=}~ \int_{(\zeta ,\nu)} \frac{\exp(\delta _{j}+u_{j})}{1+\sum_{b=1}^{J} \exp(\delta _{b}+u_{b})} f(\zeta ,\nu ;\lambda)  d(\zeta ,\nu) \\
& =~ \int_{(\zeta ,\nu)} \frac{\exp(\beta x_{j} -\alpha p_{j}+\xi _{j}+\zeta x_{j} - \nu p_{j})}{1+\sum_{b=1}^{J} \exp(\beta x_{b} -\alpha p_{b}+\xi _{b}+\zeta x_{b} -\nu p_{b})} f(\zeta ,\nu ;\lambda)  d(\zeta ,\nu),
\end{align*}
as desired, where (1) holds by $\{y=j\} = \{ \delta _{j}+u_{j}+\epsilon _{j} \geq \max_{b \in \{0,\dots ,J\} \setminus \{j\}} (\delta _{b}+u_{b}+\epsilon _{b} ) \}$, (2) by $(\zeta ,\nu) \perp (x,p,z,\xi)$ and $dP(\zeta ,\nu \mid x,p,z,\xi) = f(\zeta ,\nu ;\lambda) \, d(\zeta ,\nu)$, and (3) by the fact that, conditional on $(x,p,z,\xi,\zeta,\nu)$, $(\delta _{b}+u_{b})_{b=1}^{J}$ are non-stochastic, and $(\epsilon _{b})_{b=0}^{J}$ are i.i.d.\ Type I extreme value.
\end{proof}

\begin{lemma}\label{lem:U_is_Closed}
For any $\theta =( \alpha ,\beta ,\lambda ) \in \Theta$, $\mathcal{U}(w;\theta )$ in \eqref{eq:U_defn} is a closed set.
\end{lemma}
\begin{proof}
Consider an arbitrary sequence $\{u_{\ell}\in \mathcal{U}(w;\theta ):\ell\in \mathbb{N}\}$ with $u_{\ell}\to u \in \mathbb{R}^J$. To complete the proof, it suffices to show that $u\in \mathcal{U}(w;\theta )$.

For any $\ell\in \mathbb{N}$ and $j=1,\dots,J$, let 
\begin{align*}
A_{j,\ell} &~=~ \int_{(\zeta ,\nu )}\frac{\exp( \beta x_{j}-\alpha p_{j}+u_{m,j}+\zeta x_{j} -\nu p_{j} )}{1+\sum_{b=1}^{J}\exp( \beta x_{b}-\alpha p_{b}+u_{\ell,b}+\zeta x_{b} -\nu p_{b} )} f( \zeta ,\nu ;\lambda ) d( \zeta ,\nu ) \\
A_{0,\ell} &~=~ \int_{(\zeta ,\nu )}\frac{1}{1+\sum_{b=1}^{J}\exp( \beta x_{b}-\alpha p_{b}+u_{\ell,b}+\zeta x_{b} -\nu p_{b} )} f( \zeta ,\nu ;\lambda ) d( \zeta ,\nu ) \\
B_{\ell} &~=~ \int_{(\zeta ,\nu )}\frac{\sum_{l=1}^{J}\exp( \beta x_{l}-\alpha p_{l}+u_{\ell,l}+\zeta x_{l} -\nu p_{l} )}{1+\sum_{b=1}^{J}\exp( \beta x_{b}-\alpha p_{b}+u_{\ell,b}+\zeta x_{b} -\nu p_{b} )} f( \zeta ,\nu ;\lambda ) d( \zeta ,\nu )
\end{align*}
and
\begin{align}
A_{j} &~=~ \int_{(\zeta ,\nu )}\frac{\exp( \beta x_{j}-\alpha p_{j}+u_{j}+\zeta x_{j} -\nu p_{j} )}{1+\sum_{b=1}^{J}\exp( \beta x_{b}-\alpha p_{b}+u_{b}+\zeta x_{b} - \nu p_{b})} f( \zeta ,\nu ;\lambda ) d( \zeta ,\nu ) \notag\\
A_{0} &~=~ \int_{(\zeta ,\nu )}\frac{1}{1+\sum_{b=1}^{J}\exp( \beta x_{b}-\alpha p_{b}+u_{b}+\zeta x_{b} - \nu p_{b})} f( \zeta ,\nu ;\lambda ) d( \zeta ,\nu ) \notag\\
B &~=~\int_{(\zeta ,\nu )}\frac{\sum_{l=1}^{J}\exp( \beta x_{l}-\alpha p_{l}+u_{l}+\zeta x_{l} -\nu p_{l} )}{1+\sum_{b=1}^{J}\exp( \beta x_{b}-\alpha p_{b}+u_{b}+\zeta x_{b} -\nu p_{b} )} f( \zeta ,\nu ;\lambda ) d( \zeta ,\nu ).\label{eq:Bdefn}
\end{align}

By definition, $u_{\ell}\in \mathcal{U}(w;\theta )$ is equivalent to
\begin{equation}
A_{0,\ell} \in S_0(w)~~\text{ and }~~ \tilde{s}_{j}~=~A_{j,\ell}/B_{\ell}~~~~~\text{ for all }j=1,\dots,J.\label{eq:U_is_closed1}
\end{equation}
To complete the proof, it suffices to show that
\begin{equation}
A_{0} \in S_0(w)~~\text{ and }~~ \tilde{s}_{j}~=~A_{j}/B~~~~~\text{ for all }j=1,\dots,J.\label{eq:U_is_closed2}
\end{equation}

Fix $j=1,\dots,J$ for the remainder of this proof. Note that
\begin{align}
\lim_{\ell\to \infty }\frac{\exp( \beta x_{j}-\alpha p_{j}+u_{\ell,j}+\zeta x_{j} -\nu p_{j} )}{1+\sum_{b=1}^{J}\exp( \beta x_{b}-\alpha p_{b}+u_{\ell,b}+\zeta x_{b} -\nu p_{b} )}
&~=~ \frac{\exp( \beta x_{j}-\alpha p_{j}+u_{j}+\zeta x_{j} -\nu p_{j} )}{1+\sum_{b=1}^{J}\exp( \beta x_{b}-\alpha p_{b}+u_{b}+\zeta x_{b} -\nu p_{b} )}.\label{eq:U_is_closed3}
\end{align}
Moreover,
\begin{equation}
\frac{\exp( \beta x_{j}-\alpha p_{j}+u_{\ell,j}+\zeta x_{j} -\nu p_{j})}{1+\sum_{b=1}^{J}\exp( \beta x_{b}-\alpha p_{b}+u_{\ell,b}+\zeta x_{b} -\nu p_{b} )} ~\in~ [0,1).\label{eq:U_is_closed4}
\end{equation}
By \eqref{eq:U_is_closed3} and \eqref{eq:U_is_closed4}, the dominated convergence theorem implies that
\begin{equation}
\lim_{\ell\to \infty }A_{j,\ell}~=~A_{j}.\label{eq:U_is_closed5}
\end{equation}
By an analogous argument, we get
\begin{equation}
\lim_{\ell\to \infty }A_{0,\ell}~=~A_{0}~~\text{ and }~~ \lim_{\ell\to \infty }B_{\ell}~=~B.\label{eq:U_is_closed6}
\end{equation}
By \eqref{eq:U_is_closed1}, \eqref{eq:U_is_closed5}, and \eqref{eq:U_is_closed6}, we conclude that $A_j=B\tilde{s}_j$. By \eqref{eq:Bdefn} and $u\in\mathbb{R}^J$, we conclude that $B>0$, and so $\tilde{s}_j=A_j/B$. Since $j=1,\dots,J$ was arbitrarily chosen, this proves the second part of \eqref{eq:U_is_closed2}.

It remains to prove that $A_0\in S_0(w)$. Since $u\in\mathbb{R}^J$, for every $(\zeta,\nu)$,
\begin{equation}
\frac{1}{1+\sum_{b=1}^{J}\exp( \beta x_{b}-\alpha p_{b}+u_{b}+\zeta x_{b} -\nu p_{b} )} ~\in~ (0,1).
\label{eq:U_is_closed7}
\end{equation}
Since $f(\zeta,\nu;\lambda)d(\zeta,\nu)$ is a probability measure, \eqref{eq:U_is_closed7} implies that $A_0\in(0,1)$. By \eqref{eq:U_is_closed1}, $A_{0,\ell}\in S_0(w)$ for all $\ell\in\mathbb N$. By \eqref{eq:U_is_closed6}, $A_{0,\ell}\to A_0$. Since $A_0\in(0,1)$ and $S_0(w)$ is closed relative to $(0,1)$, we conclude that $A_0\in S_0(w)$. This proves the first part of \eqref{eq:U_is_closed2} and completes the proof.
\end{proof}

\begin{lemma}\label{lem:charaU}
For any $\theta =( \alpha ,\beta ,\lambda ) \in \Theta$, $\mathcal{U}(w;\theta )$ in \eqref{eq:U_defn} satisfies
\begin{equation}
\mathcal{U}(w;\theta )~=~\bigcup_{s_{0}\in S_0(w)}\left\{ \sigma ^{-1}\left( \tilde{s}(1-s_{0}),x,p;\lambda \right) -(\beta x_j - \alpha p_j)_{j=1}^J \right\}   \label{eq:charaU}
\end{equation}
\end{lemma}
\begin{proof}
Fix $\theta =(\alpha ,\beta ,\lambda ) \in \Theta$ arbitrarily. Let $\tilde{\mathcal{U}}$ denote the RHS of \eqref{eq:charaU}. Our goal is to show that $\mathcal{U}(w;\theta )=\tilde{\mathcal{U}}$.

\noindent \underline{Part 1.} $\mathcal{U}(w;\theta )\subseteq \tilde{\mathcal{U}}$. Take $\xi \in \mathcal{U}(w;\theta )$. First, set
\begin{equation}
s_{0}~\equiv~ \int_{(\zeta ,\nu )}\frac{1}{1+\sum_{b=1}^{J}\exp (\beta x_{b}-\alpha p_{b}+\xi _{b}+\zeta x_{b} -\nu p_{b} )}f(\zeta ,\nu ;\lambda )\,d(\zeta ,\nu )~\in~ (0,1) 
.\label{eq:charaU_1}
\end{equation}
Also, set $s\equiv \tilde{s}(1-s_{0})$, and so
\begin{equation}
\sigma ^{-1}(\tilde{s}(1-s_{0}),x,p;\lambda )-(\beta x_j - \alpha p_j)_{j=1}^J~\overset{(1)}{=}~\sigma ^{-1}(s,x,p;\lambda )-(\beta x_j - \alpha p_j)_{j=1}^J~\overset{(2)}{=}~\xi ,
\label{eq:charaU_3}
\end{equation}
where (1) holds by $s=\tilde{s}(1-s_{0})$ and (2) by \eqref{eq:sigma_defn} and \eqref{eq:charaU_1}. By $s_{0}\in S_0(w)$ and \eqref{eq:charaU_3}, we deduce that $\xi \in \tilde{\mathcal{U}}$.

\noindent \underline{Part 2.} $\mathcal{U}(w;\theta )\supseteq \tilde{\mathcal{U} }$. Take $\xi \in \tilde{\mathcal{U}}$. By definition, $\exists s_{0}\in S_0(w)$ such that 
\begin{equation}
    \sigma ^{-1}(\tilde{s}(1-s_{0}),x,p;\lambda )~=~(\beta x_j - \alpha p_j + \xi_j)_{j=1}^J,
    \label{eq:charaU_4}
\end{equation}
By applying $\sigma (\cdot ,x,p;\lambda )$ on both sides, we get 
\begin{align}
\tilde{s}(1-s_{0}) &~=~\sigma ( (\beta x_j - \alpha p_j + \xi_j)_{j=1}^J ,x,p;\lambda ) \notag  \\
&~\overset{(1)}{=}~\left\{ \int_{(\zeta ,\nu )}\frac{\exp ( \beta x_{j}-\alpha p_{j}+\xi _{j}+\zeta x_{j} -\nu p_{j} ) }{1+\sum_{b=1}^{J}\exp (\beta x_{b}-\alpha p_{b}+\xi _{b}+\zeta x_{b} -\nu p_{b} )}f(\zeta ,\nu ;\lambda )\,d(\zeta ,\nu )\right\} _{j=1}^{J},\label{eq:charaU_5}
\end{align}
where (1) holds by \eqref{eq:sigma_defn}. By \eqref{eq:charaU_5} and $\sum_{j=1}^{J}\tilde{s}_{j}=1$, we
obtain
\begin{equation}
    s_{0}~=~\int_{(\zeta ,\nu )}\frac{1}{1+\sum_{b=1}^{J}\exp (\beta x_{b}-\alpha p_{b}+\xi _{b}+\zeta x_{b} -\nu p_{b} )}f(\zeta ,\nu ;\lambda )\,d(\zeta ,\nu ).\label{eq:charaU_6}
\end{equation}
Then, we get the following derivation for any $j=1,\dots,J$, 
\begin{align*}
&\int_{(\zeta ,\nu )}\frac{\exp (\beta x_{j}-\alpha p_{j}+\xi _{j}+\zeta x_{j} -\nu p_{j} )}{1+\sum_{b=1}^{J}\exp (\beta x_{b}-\alpha p_{b}+\xi _{b}+\zeta x_{b} -\nu p_{b} )}f(\zeta ,\nu ;\lambda )\,d(\zeta ,\nu ) \\
&~\overset{(1)}{=}~\tilde{s}_j(1-s_{0}) \\
&~\overset{(2)}{=}~\tilde{s}_j\left( \int_{(\zeta ,\nu )}\frac{ \sum_{b=1}^{J}\exp (\beta x_{b}-\alpha p_{b}+\xi _{b}+\zeta x_{b} -\nu p_{b} ) }{1+\sum_{b=1}^{J}\exp (\beta x_{b}-\alpha p_{b}+\xi _{b}+\zeta x_{b} -\nu p_{b} )}f(\zeta ,\nu ;\lambda )\,d(\zeta ,\nu )\right) ,
\end{align*}
where (1) holds by \eqref{eq:charaU_5} and (2) by \eqref{eq:charaU_6}. This implies that $\xi \in \mathcal{U}(w;\theta )$, as desired.
\end{proof}

\begin{lemma}\label{lem:rand_coeff_necessary}
    Assume the conditions in Theorem \ref{thm:id_sets2}. In the plain logit case and under the assumption that $S_0(z)$ is sufficiently unrestricted (e.g., $S_{0}(z)=( 0,1)$), $\mathcal{H}(P)=\Theta_{I}(P)$.
\end{lemma}
\begin{proof}
Per the statement, we focus on the plain logit case (i.e., $\lambda =\bar{\lambda}$). By Theorem \ref{thm:id_sets1}, we have $\mathcal{H}(P)\subseteq \Theta_{I}(P)$ . Thus, it suffices to show $\Theta_{I}(P)\subseteq \mathcal{H}(P)$. Pick $\theta =\left(\beta ,\alpha ,\bar{\lambda}\right) \in \Theta_{I}(P)$ arbitrarily, and we seek to show $( \beta ,\alpha ,\bar{\lambda}) \in \Theta_{I}(P)$. By $\theta =( \beta ,\alpha ,\bar{\lambda}) \in \Theta_{I}(P)$, $\exists s^*_{0}:\mathbb{S}_{w}\rightarrow (0,1)$ with $s^*_0(w) \in S_0(w)$ for all $w \in \mathbb{S}_w$ s.t.
\begin{equation}
E\left[ \sigma ^{-1}(\tilde{s}(1-s^*_{0}(w) ),x,p;\bar{\lambda})-\left.\left( \beta x_{j}-\alpha p_{j}\right) _{j=1}^{J}\right|z \right] ~=~\mathbf{0}_{J\times 1}~\text{a.e.}~z\in \mathbb{S}_{z}.
\label{eq:id_sets_5}
\end{equation}
Under $\lambda =\bar{\lambda}$, for any $s^*_0(w) \in S_0(w)$,
\begin{equation}
\sigma ^{-1}(\tilde{s}(1-s^*_0(w)),x,p;\lambda )
~=~\left\{ \ln ( \tilde{s}_{j})+\ln \left( \frac{1-s^*_0(w) }{s^*_0(w) }\right) \right\} _{j=1}^{J},
\label{eq:id_sets_6}
\end{equation}
and so \eqref{eq:id_sets_5} is equivalent to 
\begin{equation}
\left( E\left[ \ln (\tilde{s}_{j})-\left( \beta x_{j}-\alpha p_{j}\right) |z \right] +E\left[\left. \ln \left( \frac{1-s^*_{0}(w) }{s^*_{0}(w) }\right) \right|z\right] \right) _{j=1}^{J}~=~\mathbf{0}_{J\times 1}~\text{a.e.}~z\in \mathbb{S}_{z}.
\label{eq:id_sets_7}
\end{equation}

For each $z\in \mathbb{S}_{z}$, define $s^{**}_0:\mathbb{S}_{z}\rightarrow (0,1)$ as follows:
\begin{equation}
s^{**}_0(z) ~=~\left.\left( \exp \left( E\left[ \ln \left( \frac{1-s^*_{0}(w) }{s^*_{0}(w) }\right) \right|z\right] \right) +1\right) ^{-1}.
\label{eq:id_sets_8}
\end{equation}
Provided that $S_0(z)$ is sufficiently unrestricted (e.g., $S_0(z) = (0,1)$), the right-hand side of \eqref{eq:id_sets_8} belongs to $S_0(z)$ for all $z\in \mathbb{S}_{z}$. Then, \eqref{eq:id_sets_8} implies that $s^{**}_0:\mathbb{S}_{z}\rightarrow (0,1)$ with $s^{**}_0(z) \in S_0(z)$ for all $z \in \mathbb{S}_z$ s.t.
\begin{equation}
E\left[ \left.\ln \left( \frac{1-s^*_{0}(w) }{s^*_{0}(w) }\right) \right|z\right] ~=~\ln \left( \frac{1-s^{**}_{0}(z) }{s^{**}_{0}(z) }\right) .
\label{eq:id_sets_9}
\end{equation}
By combining \eqref{eq:id_sets_7}, and \eqref{eq:id_sets_9}, we conclude that
\begin{equation*}
\left( E\left[ \ln (\tilde{s}_{j})-\left( \beta x_{j}-\alpha p_{j}\right) |z \right] +\ln \left( \frac{1-s^{**}_{0}(z) }{{s}^{**}_{0}(z) }\right) \right) _{j=1}^{J}~=~\mathbf{0}_{J\times 1}~\text{a.e.}~z\in \mathbb{S}_{z}.
\end{equation*}
By combining this and \eqref{eq:id_sets_6}, we get
\begin{equation}
E\left[ \sigma ^{-1}(\tilde{s}(1-s^{**}_{0}(z) ),x,p;\bar{\lambda})-\left.\left( \beta x_{j}-\alpha p_{j}\right) _{j=1}^{J}\right|z\right] ~=~\mathbf{0}_{J\times 1}~\text{a.e.}~z\in  \mathbb{S}_{z}.
\end{equation}
By definition, this means that $( \beta ,\alpha ,\bar{\lambda}) \in \mathcal{H}(P).$
\end{proof}

\begin{proof}[Proof of Corollary \ref{cor:derived}]
    This follows immediately from \eqref{eq:CS_validity} and the identified sets in Theorem \ref{thm:eq_objects}.
\end{proof}

\subsection{Additional numerical results}\label{sec:Numerical_NoInfo}

This section reports additional numerical results for the case $S_0(w)=(0,1)$, which represents a complete lack of information about the share of outside goods. We use exactly the same computational procedure as in Section \ref{sec:Numerical1}, except that we replace the lower bound $c=0.8$ with $c=0$. 

This exercise has two related goals. First, it illustrates what can be learned when the researcher has no information about the outside-good share beyond the restrictions imposed by the parameter space. Second, by comparing these results with the partial-information case in the main text, it shows the identifying content of imposing a lower bound on $s_0$ in settings where, as in this design, the distribution of the outside-good share is concentrated at very high values.

Figure \ref{fig:idset_global} reports the resulting identified set and compares it with the partial-information case $S_0(w)=[0.8,1)$ discussed in the main text. The numerical results show a substantial enlargement of the identified set relative to the partial-information case. Most notably, the upper bound on $\beta_1$ is no longer informative within the parameter space considered, and, for larger values of $\lambda$, the identified sets for $\alpha$ become much wider. This comparison highlights the identifying content of the lower bound on $s_0$: once this restriction is removed, we lose a substantial amount of information about the structural parameters.

\begin{sidewaysfigure}
    \centering
    \includegraphics[width=\textheight]{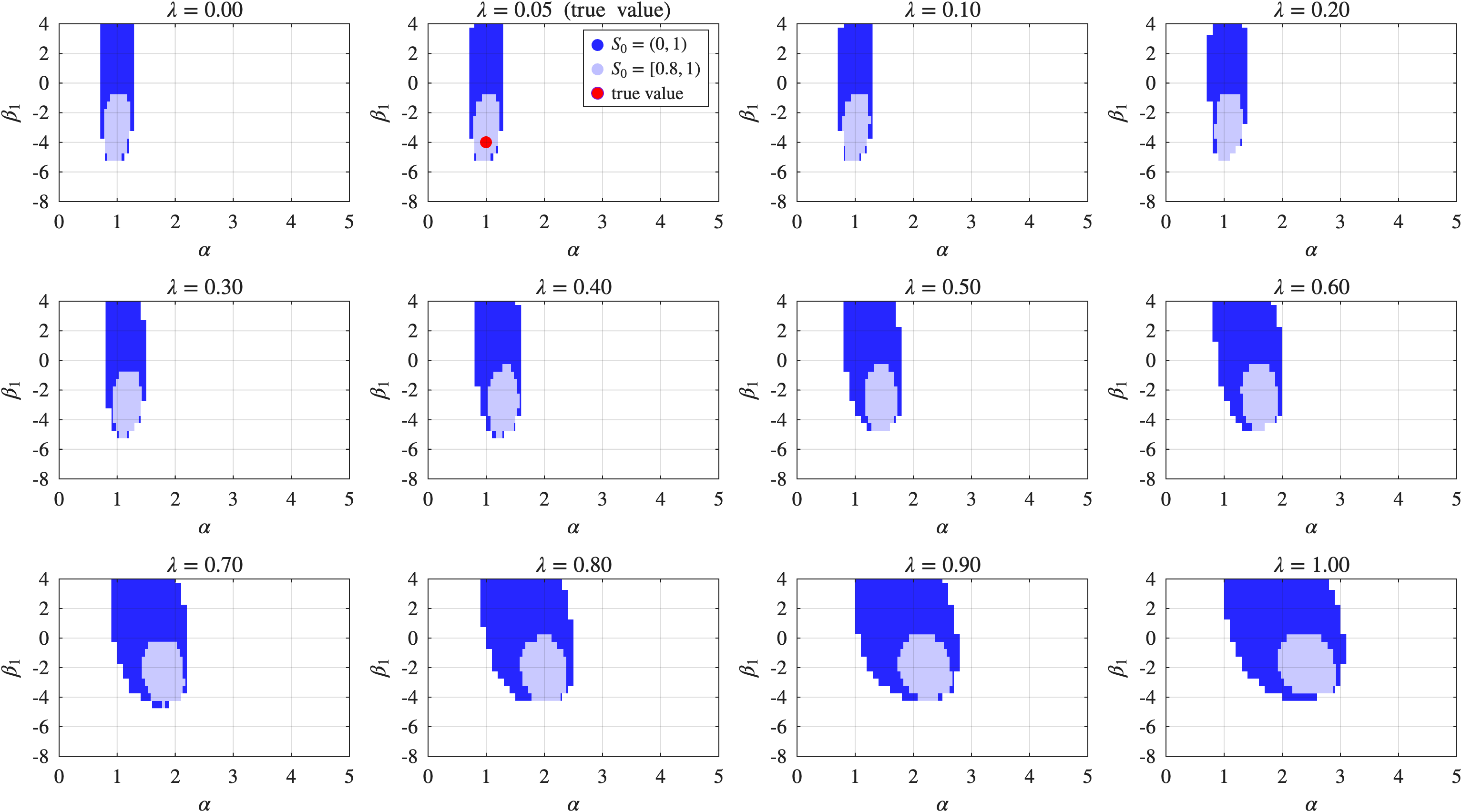}
        \caption{Identified set $\Theta _{I}(P)$ under two information structures for the outside-good share: $S_{0} = (0,1)$ (dark blue) and $S_0(w) = [0.8, 1)$ (light blue), in a DGP with $(\alpha,\beta_{1},\lambda) = (1,-4,0.05)$, computed with tolerance parameter $\tau=0.01$.}
    \label{fig:idset_global}
\end{sidewaysfigure}
\end{small}

\bibliography{BIBLIOGRAPHY}

\begin{thebibliography}{30}
\newcommand{\enquote}[1]{``#1''}
\expandafter\ifx\csname natexlab\endcsname\relax\def\natexlab#1{#1}\fi

\bibitem[\protect\citeauthoryear{Aguiar and Kashaev}{Aguiar and
  Kashaev}{2025}]{Aguiar:2025aa}
\textsc{Aguiar, V.~H. and N.~Kashaev} (2025): \enquote{Identification and
  estimation of discrete choice models with unobserved choice sets,}
  \emph{Journal of Business \& Economic Statistics}, 43, 204--215.

\bibitem[\protect\citeauthoryear{Andrews and Shi}{Andrews and
  Shi}{2013}]{andrews/shi:2013}
\textsc{Andrews, D. W.~K. and X.~Shi} (2013): \enquote{Inference Based on
  Conditional Moment Inequalities,} \emph{Econometrica}, 81, 609--666.

\bibitem[\protect\citeauthoryear{Armstrong}{Armstrong}{2016}]{armstrong:2016}
\textsc{Armstrong, T.~B.} (2016): \enquote{Large Market Asymptotics For
  Differentiated Product Demand Estimators with Economic Models of Supply,}
  \emph{Econometrica}, 84, 1961--1980.

\bibitem[\protect\citeauthoryear{Backus, Conlon, and Sinkinson}{Backus
  et~al.}{2021}]{Backus:2021aa}
\textsc{Backus, M., C.~Conlon, and M.~Sinkinson} (2021): \enquote{Common
  ownership and competition in the ready-to-eat cereal industry,} Tech. rep.,
  National Bureau of Economic Research.

\bibitem[\protect\citeauthoryear{Belloni, Bugni, and Chernozhukov}{Belloni
  et~al.}{2019}]{belloni/bugni/chernozhukov:2019}
\textsc{Belloni, A., F.~Bugni, and V.~Chernozhukov} (2019): \enquote{Subvector
  inference in partially identified moment inequality models with many moment
  inequalities,} Mimeo: Duke University and M.I.T.

\bibitem[\protect\citeauthoryear{Berry, Carnall, and Spiller}{Berry
  et~al.}{2006}]{Berry:2006tt}
\textsc{Berry, S., M.~Carnall, and P.~T. Spiller} (2006): \enquote{Airline
  hubs: costs, markups and the implications of customer heterogeneity,}
  \emph{Competition policy and antitrust}.

\bibitem[\protect\citeauthoryear{Berry, Gandhi, and Haile}{Berry
  et~al.}{2013}]{berry/gandhi/haile:2013}
\textsc{Berry, S., A.~Gandhi, and P.~Haile} (2013): \enquote{Connected
  Substitutes and Invertibility of Demand,} \emph{Econometrica}, 81,
  2087--2111.

\bibitem[\protect\citeauthoryear{Berry and Haile}{Berry and
  Haile}{2014}]{berry/haile:2014}
\textsc{Berry, S. and P.~Haile} (2014): \enquote{Identification in
  Differentiated Products Markets Using Market Level Data,}
  \emph{Econometrica}, 82, 1749--1797.

\bibitem[\protect\citeauthoryear{Berry, Levinson, and Pakes}{Berry
  et~al.}{1995}]{berry/levinsohn/pakes:1995}
\textsc{Berry, S., J.~Levinson, and A.~Pakes} (1995): \enquote{Automobile
  Prices in Market Equilibrium,} \emph{Econometrica}, 63, 841--890.

\bibitem[\protect\citeauthoryear{Berry, Linton, and Pakes}{Berry
  et~al.}{2004}]{berry/linton/pakes:2004}
\textsc{Berry, S., O.~Linton, and A.~Pakes} (2004): \enquote{Limit Theorems for
  Estimating the Parameters of Differentiated Product Demand Systems,}
  \emph{Review of Economic Studies}, 71, 613--654.

\bibitem[\protect\citeauthoryear{Berry}{Berry}{1994}]{berry:1994}
\textsc{Berry, S.~T.} (1994): \enquote{Estimating Discrete-Choice Models of
  Product Differentiation,} \emph{The RAND Journal of Economics}, 25, 242--262.

\bibitem[\protect\citeauthoryear{Berry and Haile}{Berry and
  Haile}{2021}]{berry/haile:2021}
\textsc{Berry, S.~T. and P.~A. Haile} (2021): \enquote{Foundations of demand
  estimation,} in \emph{Handbook of Industrial Organization}, Elsevier, vol.~4,
  1--62.

\bibitem[\protect\citeauthoryear{Billingsley}{Billingsley}{1995}]{billingsley:1995}
\textsc{Billingsley, P.} (1995): \emph{Probability and Measure}, John Wiley and
  Sons, Inc.

\bibitem[\protect\citeauthoryear{Byrne, Imai, Jain, and Sarafidis}{Byrne
  et~al.}{2022}]{Byrne:2022vz}
\textsc{Byrne, D.~P., S.~Imai, N.~Jain, and V.~Sarafidis} (2022):
  \enquote{Instrument-free identification and estimation of differentiated
  products models using cost data,} \emph{Journal of Econometrics}, 228,
  278--301.

\bibitem[\protect\citeauthoryear{Chernozhukov, Chetverikov, and
  Kato}{Chernozhukov et~al.}{2019}]{chernozhukov/chetverikov/kato:2019}
\textsc{Chernozhukov, V., D.~Chetverikov, and K.~Kato} (2019):
  \enquote{Inference on Causal and Structural Parameters using Many Moment
  Inequalities,} \emph{The Review of Economic Studies}, 86, 1867--1900.

\bibitem[\protect\citeauthoryear{Chesher and Rosen}{Chesher and
  Rosen}{2017}]{chesher/rosen:2017}
\textsc{Chesher, A. and A.~Rosen} (2017): \enquote{Generalized Instrumental
  Variable Models,} \emph{Econometrica}, 85, 959--989.

\bibitem[\protect\citeauthoryear{Chu, Leslie, and Sorensen}{Chu
  et~al.}{2011}]{Chu:2011ul}
\textsc{Chu, C.~S., P.~Leslie, and A.~Sorensen} (2011): \enquote{Bundle-size
  pricing as an approximation to mixed bundling,} \emph{The American Economic
  Review}, 263--303.

\bibitem[\protect\citeauthoryear{Freyberger}{Freyberger}{2015}]{freyberger:2015}
\textsc{Freyberger, J.} (2015): \enquote{Asymptotic theory for differentiated
  products demand models with many markets,} \emph{Journal of Econometrics},
  185, 162--181.

\bibitem[\protect\citeauthoryear{Gandhi, Lu, and Shi}{Gandhi
  et~al.}{2023}]{gandhi/lu/shi:2023}
\textsc{Gandhi, A., Z.~Lu, and X.~Shi} (2023): \enquote{Estimating Demand for
  Differentiated Products with Zeroes in Market Share Data,} \emph{Quantitative
  Economics}, 14, 381--418.

\bibitem[\protect\citeauthoryear{Hong, Li, and Li}{Hong
  et~al.}{2021}]{hong/li/li:2021}
\textsc{Hong, H., H.~Li, and J.~Li} (2021): \enquote{BLP Estimation Using
  Laplace Transformation and Overlapping Simulation Draws,} \emph{Journal of
  Econometrics}, 222, 56--72.

\bibitem[\protect\citeauthoryear{Horta{\c{c}}su, Oery, and
  Williams}{Horta{\c{c}}su et~al.}{2022}]{Hortacsu:2022wt}
\textsc{Horta{\c{c}}su, A., A.~Oery, and K.~R. Williams} (2022):
  \enquote{Dynamic price competition: Theory and evidence from airline
  markets,} Tech. rep., National Bureau of Economic Research.

\bibitem[\protect\citeauthoryear{Huang and Rojas}{Huang and
  Rojas}{2013}]{Huang:2013vh}
\textsc{Huang, D. and C.~Rojas} (2013): \enquote{The Outside Good Bias in Logit
  Models of Demand with Aggregate Data,} \emph{Economics Bulletin}, 33,
  198--206.

\bibitem[\protect\citeauthoryear{Kim and Sawada}{Kim and
  Sawada}{2024}]{Kim:2024ti}
\textsc{Kim, D. and M.~Sawada} (2024): \enquote{Market Size and Market Power,}
  \emph{Available at SSRN 4719046}.

\bibitem[\protect\citeauthoryear{Li, Mazur, Park, Roberts, Sweeting, and
  Zhang}{Li et~al.}{2022}]{Li:2022tk}
\textsc{Li, S., J.~Mazur, Y.~Park, J.~Roberts, A.~Sweeting, and J.~Zhang}
  (2022): \enquote{Repositioning and market power after airline mergers,}
  \emph{The RAND Journal of Economics}.

\bibitem[\protect\citeauthoryear{Molchanov}{Molchanov}{2005}]{molchanov:2005}
\textsc{Molchanov, I.} (2005): \emph{Theory of Random Sets}, Probability and
  Its Applications, London: Springer.

\bibitem[\protect\citeauthoryear{Molchanov and Molinari}{Molchanov and
  Molinari}{2018}]{molchanov/molinari:2018}
\textsc{Molchanov, I. and F.~Molinari} (2018): \emph{Random Sets in
  Econometrics}, vol.~60 of \emph{Econometric Society Monographs}, Cambridge
  University Press.

\bibitem[\protect\citeauthoryear{Moon, Shum, and Weidner}{Moon
  et~al.}{2018}]{Moon:2018aa}
\textsc{Moon, H.~R., M.~Shum, and M.~Weidner} (2018): \enquote{Estimation of
  random coefficients logit demand models with interactive fixed effects,}
  \emph{Journal of Econometrics}, 206, 613--644.

\bibitem[\protect\citeauthoryear{Nevo and Hatzitaskos}{Nevo and
  Hatzitaskos}{2006}]{Nevo:2006aa}
\textsc{Nevo, A. and K.~Hatzitaskos} (2006): \enquote{Why does the average
  price paid fall during high demand periods?} Tech. rep., CSIO working paper.

\bibitem[\protect\citeauthoryear{Sweeting, Roberts, and Gedge}{Sweeting
  et~al.}{2020}]{Sweeting:2020tw}
\textsc{Sweeting, A., J.~W. Roberts, and C.~Gedge} (2020): \enquote{A model of
  dynamic limit pricing with an application to the airline industry,}
  \emph{Journal of Political Economy}, 128, 1148--1193.

\bibitem[\protect\citeauthoryear{Zhang}{Zhang}{2024}]{Zhang:2024}
\textsc{Zhang, L.} (2024): \enquote{Identification and estimation of market
  size in discrete choice demand models,} SSRN 4545848.

\end{thebibliography}
\end{document}